\documentclass[11pt]{article}
\usepackage{graphicx,amsmath,amsfonts,amssymb,bbm,amsthm}
\usepackage{psfrag}
\usepackage{epsf}
\newcommand{\pt}{{\cal PT}}

\newcommand{\re}{\mathrm{Re}}
\newcommand{\im}{\mathrm{Im}}

\newcommand{\overbar}[1]{\mkern 1.5mu\overline{\mkern-1.5mu#1\mkern-1.5mu}\mkern 1.5mu}

\newtheorem{theorem}{Theorem}[section]
\newtheorem{lemma}[theorem]{Lemma}
\newtheorem{proposition}[theorem]{Proposition}
\newtheorem{corollary}[theorem]{Corollary}
\newtheorem{remark}{Remark}[section]

\begin{document}

\title{\bf Existence and stability of $\pt$-symmetric vortices \\
in nonlinear two-dimensional square lattices}

\author{Haitao Xu$^{1}$, P.G. Kevrekidis$^{1}$ and Dmitry E. Pelinovsky$^{2,3}$ \\
{\small $^{1}$ Department of Mathematics and Statistics, University of
Massachusetts, Amherst, MA 01003-9305, USA} \\
{\small $^{2}$ Department of Mathematics, McMaster
University, Hamilton, Ontario, Canada, L8S 4K1} \\
{\small $^{3}$ Department of Applied Mathematics, Nizhny Novgorod State Technical University, Nizhny Novgorod, Russia } }

\date{\today}
\maketitle

\begin{abstract}
Vortices symmetric with respect to simultaneous parity and time reversing transformations
are considered on the square lattice in the framework of the discrete nonlinear Schr\"{o}dinger
equation. The existence and stability of vortex configurations is analyzed in the limit
of weak coupling between the lattice sites, when predictions on the elementary cell of a square
lattice (i.e., a single square) can be extended to a large (yet finite) array of lattice cells.
Our analytical predictions are found to be in good agreement with numerical computations.
\end{abstract}

\section{Introduction}

Networks of coupled nonlinear oscillators with balanced gains and losses have been
considered recently in the context of nonlinear $\pt$-symmetric lattices. Among many other questions, attention has been paid to issues of
linear and nonlinear stability of constant equilibrium states \cite{pel1} and spatially distributed steady states \cite{pel2,pel3} in such systems.
More generally, the study of solitary waves in such $\pt$-symmetric
lattices~\cite{dmitriev} has garnered considerable attention over the years,
as can be inferred also from a recent comprehensive review
on the subject~\cite{yura}. A significant recent boost to
the relevant interest has been offered by the experimental observation
of optical solitons in lattice settings~\cite{miri}.
Many of the relevant theoretical notions have been developed also
in continuum systems with periodic potentials in both
scalar~\cite{mussli1,mussli2} and vector~\cite{kartash} settings.

In higher dimensions, the number of studies of $\pt$-symmetric
lattices and the coherent structures that they support
is considerably more limited.
Nonlinear states bifurcating out of linear (point spectrum)
modes of a potential and their stability have been studied~\cite{jianke}
and so have gap solitons~\cite{zhu}. However, an understanding
of fundamentally topological states such as vortices and their
existence and stability properties is still an active theme of study. The few
studies addressing these topological structures
have been chiefly numerical in
nature~\cite{malomchina,leykam,guenther}.
It is, thus, the aim of the present study to explore a two-dimensional
square lattice setting and to provide an understanding of the existence and
stability properties
of the stationary states it can support, placing a particular emphasis on
the vortical structures.

We will be particularly interested in the following
$\pt$-symmetric model of the discrete nonlinear Schr\"{o}dinger (dNLS) type~\cite{dnlsbook},
\begin{equation}
\label{pt-dnls}
i \frac{d \psi_j}{d t} + (\Delta \psi)_j +|\psi_j|^2 \psi_j = i \gamma_j \psi_j,
\end{equation}
where $\psi_j \in \mathbb{C}$ depends on the lattice site $j \in \mathbb{Z}^2$ and
the time variable $t \in \mathbb{R}$ (in optics, this corresponds
to the spatial propagation direction),
$(\Delta \psi)_j$ denotes the discrete Laplacian operator at the $j$-th site of the square lattice,
and the distribution of the parameter values $\gamma_j \in \mathbb{R}$ for gains
or losses is supposed to be $\pt$-symmetric. The dNLS equation
is a prototypical model for the study of optical waveguide
arrays~\cite{moti}, and the principal setting in which $\pt$-symmetry
was first developed experimentally (in the context of few waveguides,
such as the dimer setting~\cite{kip}).

The $\pt$-symmetry holds if
the distribution of $\gamma_j$ is odd with respect to reflections
of the lattice sites in $\mathbb{Z}^2$ about a selected center or line
of symmetry. In particular, we will consider two
natural symmetric configurations shown in Figure \ref{vortex1}. The left panel
shows the symmetry about a vertical line located on the equal distance
between two vertical arrays of lattice sites. The right panel shows the symmetry
about a center point in the elementary
cell of the square lattice.
\begin{figure}[!htbp]
\centering
\begin{tabular}{c}
\includegraphics[width=12cm]{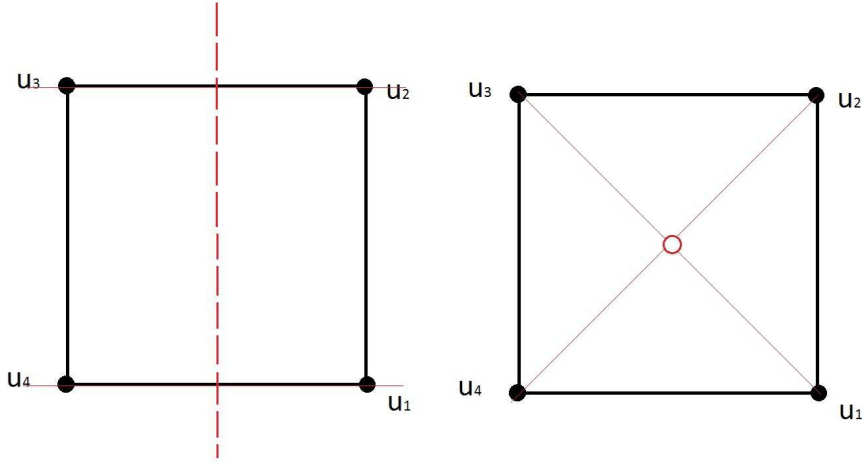}
\end{tabular}
\caption{Schematic $\pt$-symmetry for the elementary cell of the square lattice.}
\label{vortex1}
\end{figure}

In order to enable the analytical consideration of the existence and stability of vortices
in the $\pt$-symmetric dNLS equation (\ref{pt-dnls}), we will consider steady states in the limit of
large energy \cite{pel2,pel3}. This enables a formalism of the so-called anti-continuum limit
in analysis of steady states in nonlinear lattices. In particular, we
set $\psi_j(t) = \varphi_j(t) e^{iEt}$ and introduce the scaling
\begin{equation}
\label{transformation}
E = \epsilon^{-1}, \quad \varphi_j(t) = u_j(\tau) \epsilon^{-1/2}, \quad \tau = t \epsilon^{-1}.
\end{equation}
As a result of this transformation, the $\pt$-symmetric dNLS equation (\ref{pt-dnls}) can be rewritten in the equivalent form
\begin{equation}
\label{pt-dnls-equiv}
i \frac{d u_j}{d \tau} - u_j + \epsilon (\Delta u)_j + |u_j|^2 u_j = i \epsilon \gamma_j u_j,
\end{equation}
where the parameter $\epsilon$ is small in the limit of large energy $E$.
Moreover, if $\epsilon \to 0^+$, then $E \to +\infty$, whereas if $ \epsilon \to 0^-$, then $E \to -\infty$.
Setting $\epsilon = 0$ yields the system
of uncoupled conservative nonlinear oscillators, therefore, small values of $\epsilon$ can be considered by
the perturbation theory from the uncoupled conservative limit.
In light of this transformation, the analysis will
follow our previous work on existence
and stability of vortices in conservative lattices of the dNLS type \cite{kev1} (see also applications in
\cite{kev2,kev3}). In what follows, we apply the continuation technique to obtain definite conclusions
on vortices in the $\pt$-symmetric dNLS equation (\ref{pt-dnls-equiv}).

We will focus on the basic vortex configuration, for which the excited oscillators
are only supported on the elementary cell shown on Figure \ref{vortex1}.
In this case, the definite conclusions on existence of $\pt$-symmetric vortices
can already be extracted from studies of the $\pt$-symmetric dNLS equation (\ref{pt-dnls-equiv})
on four sites only, see Section \ref{section-four-sites}.
This is the so-called plaquette setting in \cite{guenther}.
With the $\pt$-symmetry in hand and appropriate (e.g. Dirichlet)
boundary conditions on the square lattice truncated symmetrically in a suitable square domain,
one can easily upgrade these conclusions for the full dNLS equation (\ref{pt-dnls-equiv}),
see Section \ref{section-full}. The existence results
remain valid in the infinite square lattice, thanks to the choice of the sequence spaces
such as $\ell^2(\mathbb{Z}^2)$.

Stability of $\pt$-symmetric vortices on the four sites can be analyzed in the framework
of the Lyapunov--Schmidt reduction method, see Section \ref{section-stability}.
However, one needs to be more careful to study stability of the localized steady states
on large square lattices because the zero equilibrium may become spectrally unstable
in the lattices with spatially extended gains and losses \cite{lett1,lett2}.
Stable configurations depend sensitively on the way gains and losses compensate each other,
especially if the two-dimensional square lattice is truncated to a finite size.
These aspects are discussed in Section \ref{section-stability-extended}.
Finally, Section~\ref{conclu} provides some conclusions, as well as
offers an outlook towards future work.

\section{Existence of $\pt$-symmetric vortices in the elementary cell}
\label{section-four-sites}

Let us consider the elementary cell of the square lattice shown on Figure \ref{vortex1}. We will enumerate the four corner sites
in the counterclockwise order with the first site to lie at the bottom right. By the construction,
the configuration has a cyclic symmetry with respect to the shift along the elementary cell.

Looking for the steady-state
solutions $u_j(\tau) = \phi_j e^{-2i \epsilon \tau}$, where the exponential factor removes the diagonal
part of the Laplacian operator $\Delta$ connecting each of the three lattice sites in the elementary cell,
we rewrite the $\pt$-symmetric dNLS equation (\ref{pt-dnls-equiv}) in the explicit form
\begin{eqnarray}
\label{eqn0_1}
\left\{ \begin{array}{l}
(1-|\phi_1|^2)\phi_1 - \epsilon( \phi_2 + \phi_4 - i\gamma_1 \phi_1) = 0, \\
(1-|\phi_2|^2)\phi_2 - \epsilon( \phi_1 + \phi_3 - i\gamma_2 \phi_2) = 0, \\
(1-|\phi_3|^2)\phi_3 - \epsilon( \phi_2 + \phi_4 - i\gamma_3 \phi_3) = 0, \\
(1-|\phi_4|^2)\phi_4 - \epsilon( \phi_1 + \phi_3 - i\gamma_4 \phi_4) = 0. \end{array} \right.
\end{eqnarray}

In the following, we consider two types of $\pt$-symmetric configurations
for the gain and loss parameters. These two configurations
correspond to the left and right panels of Figure \ref{vortex1}.
\begin{itemize}
\item[(S1)] Symmetry about the vertical line: $\gamma_1=-\gamma_4$ and $\gamma_2=-\gamma_3$;
\item[(S2)] Symmetry about the center: $\gamma_1=-\gamma_3$ and $\gamma_2=-\gamma_4$.
\end{itemize}
Besides the cyclic symmetry, one can also flip each of the two configurations
about the vertical or horizontal axes of symmetries. In addition, for the configuration
shown on the right panel of Figure \ref{vortex1}, one can also flip the configuration
about the center of symmetry, either between the first and third sites or between the second and fourth sites.

The $\pt$-symmetric stationary
states that we explore correspond to particular reductions of the system of
algebraic equations (\ref{eqn0_1}), namely,
\begin{itemize}
\item[(S1)] Symmetry about the vertical line: $\phi_1=\bar{\phi}_4$ and $\phi_2=\bar{\phi}_3$;
\item[(S2)] Symmetry about the center: $\phi_1=\bar{\phi}_3$ and $\phi_2=\bar{\phi}_4$.
\end{itemize}
Due to the symmetry conditions, the system of algebraic equations (\ref{eqn0_1})
reduces to two equations for $\phi_1$ and $\phi_2$ only. We shall classify all possible solutions separately
for the two different symmetries. We also note the symmetry of solutions with respect to changes in the sign of
$\{ \phi_j \}_{1 \leq j \leq 4}$, $\{ \gamma_j \}_{1 \leq j \leq 4}$, and $\epsilon$.
\begin{remark}
\label{remark1_1}
If $\{\phi_j\}_{1\leq j\leq4}$ solves the system (\ref{eqn0_1}) with $\{\gamma_j\}_{1\leq j\leq4}$ and $\epsilon$,
then $\{\overbar{\phi}_j\}_{1\leq j\leq4}$ solves the same system with $\{-\gamma_j\}_{1\leq j\leq4}$ and $\epsilon$.
\end{remark}

\begin{remark}
\label{remark1_2}
If $\{\phi_j\}_{1\leq j\leq4}$ solves the system (\ref{eqn0_1}) with $\{\gamma_j\}_{1\leq j\leq4}$ and $\epsilon$,
then $\{ (-1)^j \overbar{\phi}_j\}_{1\leq j\leq4}$ solves the same system with $\{\gamma_j\}_{1\leq j\leq4}$ and $-\epsilon$.
\end{remark}

\begin{remark}
\label{remark1_3}
If $\{\phi_j\}_{1\leq j\leq4}$ solves the system (\ref{eqn0_1}) with $\{\gamma_j\}_{1\leq j\leq4}$ and $\epsilon$,
then $\{-\phi_j\}_{1\leq j\leq4}$ solves the same system with $\{\gamma_j\}_{1\leq j\leq4}$ and $\epsilon$.
\end{remark}

As is known from the previous work \cite{kev1}, if gains and losses are absent, that is, if $\gamma_j = 0$ for all $j$,
then the solutions of the system (\ref{eqn0_1}) are classified into two groups:
\begin{itemize}
\item discrete solitons if $\arg(\phi_j) = \theta_0 \; {\rm mod}(\pi)$ for all $j$;
\item discrete vortices, otherwise.
\end{itemize}
Persistence of discrete solitons is well known for the $\pt$-symmetric networks \cite{pel2,pel3},
whereas persistence of discrete vortices has not been theoretically
established in the
literature. The term ``persistence"
refers to the unique continuation of the limiting configuration at $\epsilon = 0$ with respect to the small
parameter $\epsilon$. The gain and loss parameters are considered to be fixed in this continuation.

\subsection{Symmetry about the vertical line (S1)}

Under conditions $\gamma_1=-\gamma_4$, $\gamma_2=-\gamma_3$, $\phi_1=\bar{\phi}_4$, and
$\phi_2=\bar{\phi}_3$, the system (\ref{eqn0_1}) reduces to two algebraic equations:
\begin{eqnarray}
\label{eqn1_1}
\left\{ \begin{array}{l}
f_1 := (1-|\phi_1|^2)\phi_1 - \epsilon( \bar{\phi}_1 + \phi_2 - i\gamma_1 \phi_1 ) = 0, \\
f_2 := (1-|\phi_2|^2)\phi_2 - \epsilon( \bar{\phi}_2 + \phi_1 - i\gamma_2 \phi_2 ) = 0. \end{array} \right.
\end{eqnarray}
In general, it is not easy to solve the system (\ref{eqn1_1}) for any given $\gamma_1$, $\gamma_2$ and $\epsilon$.
However, branches of solutions can be classified through continuation from the limiting case $\epsilon=0$
to an open set $\mathcal{O}(0)$ of the $\epsilon$ values that contains $0$. Simplifying the general approach
described in \cite{kev1} for the $\pt$-symmetric vortex configurations, we obtain the following result.

\begin{lemma}
\label{lemma2_1}
Consider the general solution of the system (\ref{eqn1_1}) at $\epsilon = 0$ in the form:
\begin{equation}
\label{limit-vortex}
\phi_j^{(\epsilon = 0)}(\theta_j) = e^{i \theta_j}, \quad \theta_j \in \mathbb{T} := \mathbb{R} / (2 \pi \mathbb{Z}).
\end{equation}
For every $\gamma_1$ and $\gamma_2$, there exists a $C^{\infty}$ function
$\boldsymbol{h}(\boldsymbol{\theta},\epsilon) : \mathbb{T}^2 \times \mathbb{R} \to \mathbb{R}^2$
such that there exists a unique solution $\boldsymbol{\phi} \in \mathbb{C}^2$
to the system (\ref{eqn1_1}) near $\boldsymbol{\phi}^{(\epsilon = 0)}(\boldsymbol{\theta}) \in \mathbb{C}^2$
for every $\epsilon \in \mathcal{O}(0)$ if and only if there exists a unique solution
$\boldsymbol{\theta} \in \mathbb{T}^2$ of the system $\boldsymbol{h}(\boldsymbol{\theta},\epsilon)=0$ for
every $\epsilon \in \mathcal{O}(0)$.
\end{lemma}

\begin{proof}
Representing the unknown solution with $\phi_j = r_j e^{i \theta_j}$, where $r_j \in \mathbb{R}^+$ and $\theta_j \in \mathbb{T}$,
we separate the real and imaginary parts in the form $g_j := {\rm Re}(f_j e^{-i \theta_j})$ and
$h_j := {\rm Im}(f_j e^{-i \theta_j})$. For convenience, we write the explicit expressions:
\begin{eqnarray}
\label{eqn1_2}
\left\{ \begin{array}{l}
g_1 := (1-r_1^2) r_1 - \epsilon \left( r_1 \cos(2 \theta_1) + r_2 \cos(\theta_2 - \theta_1)\right), \\
g_2 := (1-r_2^2) r_2 - \epsilon \left( r_2 \cos(2 \theta_2) + r_1 \cos(\theta_1 - \theta_2)\right) \end{array} \right.
\end{eqnarray}
and
\begin{eqnarray}
\label{eqn1_3}
\left\{ \begin{array}{l}
h_1 := \epsilon  \left( r_1 \sin(2 \theta_1) - r_2 \sin(\theta_2 - \theta_1) + \gamma_1 r_1 \right), \\
h_2 := \epsilon \left( r_2 \sin(2 \theta_2) - r_1 \sin(\theta_1 - \theta_2) + \gamma_2 r_2 \right). \end{array} \right.
\end{eqnarray}
It is clear that $\boldsymbol{\phi}$ is a root of $\boldsymbol{f}$ if and only if
$(\boldsymbol{r},\boldsymbol{\theta}) \in \mathbb{R}^2 \times \mathbb{T}^2$
is a root of $(\boldsymbol{g},\boldsymbol{h}) \in \mathbb{R}^2 \times \mathbb{R}^2$.
Moreover, $(\boldsymbol{g},\boldsymbol{h})$ is smooth both in $(\boldsymbol{r},\boldsymbol{\theta})$ and $\epsilon$.

For $\epsilon = 0$, we pick the solution with $\boldsymbol{r} = {\bf 1}$ and
$\boldsymbol{\theta} \in \mathbb{T}^2$ arbitrary, as per the explicit expression (\ref{limit-vortex}).
Since $\boldsymbol{g}$ is smooth in $\boldsymbol{r}$,
$\boldsymbol{\theta}$, and $\epsilon$, whereas the Jacobian $\partial_{\boldsymbol{u}} \boldsymbol{g}$
at $\boldsymbol{r} = {\bf 1}$ and $\epsilon = 0$ is invertible, the Implicit Function Theorem for smooth vector functions
applies. From this theorem, we deduce the existence of a unique $\boldsymbol{r} \in \mathbb{R}^2$ near
$\boldsymbol{1} \in \mathbb{R}^2$
for every $\boldsymbol{\theta} \in \mathbb{T}^2$ and $\epsilon \in \mathbb{R}$ sufficiently small,
such that the mapping $(\boldsymbol{\theta},\epsilon) \mapsto \boldsymbol{r}$  is smooth
and $\| \boldsymbol{r} - \boldsymbol{1} \| \leq C |\epsilon|$ for an $\epsilon$-independent constant $C > 0$.

Substituting the smooth mapping $(\boldsymbol{\theta},\epsilon) \mapsto \boldsymbol{r}$
into the definition of $\boldsymbol{h}$, we obtain the smooth function
$\boldsymbol{h}(\boldsymbol{\theta},\epsilon) : \mathbb{T}^2 \times \mathbb{R} \to \mathbb{R}^2$,
the root of which yields the assertion of the lemma.
\end{proof}

Lemma \ref{lemma2_1} represents the first step of the Lyapunov--Schmidt reduction algorithm, namely,
a reduction of the original system (\ref{eqn1_1}) to the bifurcation equation for the root of
a smooth function $\boldsymbol{h}(\boldsymbol{\theta},\epsilon) : \mathbb{T}^2 \times \mathbb{R} \to \mathbb{R}^2$,
defined from the system (\ref{eqn1_2}) and (\ref{eqn1_3}). The following lemma
represents the second step of the Lyapunov--Schmidt reduction algorithm, namely,
a solution of the bifurcation equation in the same limit of small $\epsilon$.

\begin{lemma}
\label{lemma2_2}
Denote $\boldsymbol{H}(\boldsymbol{\theta}) = \lim_{\epsilon \to 0} \epsilon^{-1} \boldsymbol{h}(\boldsymbol{\theta},\epsilon)$
and the corresponding Jacobian matrix $\mathcal{N}(\boldsymbol{\theta}) = \partial_{\boldsymbol{\theta}} \boldsymbol{H}(\boldsymbol{\theta})$.
Assume that $\boldsymbol{\theta}^{(\epsilon = 0)} \in \mathbb{T}^2$ is a root of $\boldsymbol{H}$
such that $\mathcal{N}(\boldsymbol{\theta}^{(\epsilon = 0)})$ is invertible.
Then, there exists a unique root $\boldsymbol{\theta} \in \mathbb{T}^2$
of $\boldsymbol{h}(\boldsymbol{\theta},\epsilon)$ near $\boldsymbol{\theta}^{(\epsilon = 0)}$
for every $\epsilon \in \mathcal{O}(0)$ such that the mapping $\epsilon \mapsto \boldsymbol{\theta}$  is smooth
and $\| \boldsymbol{\theta} - \boldsymbol{\theta}^{(\epsilon = 0)} \| \leq C |\epsilon|$
for an $\epsilon$-independent positive constant $C$.
\end{lemma}

\begin{proof}
The particular form in the definition of $\boldsymbol{H}$ relies on the explicit definition (\ref{eqn1_3}).
The assertion of the lemma follows from the Implicit Function Theorem for smooth vector functions.
\end{proof}

\begin{corollary}
\label{cor-2-1}
Under conditions of Lemmas \ref{lemma2_1} and \ref{lemma2_2}, there exists
a unique solution $\boldsymbol{\phi} \in \mathbb{C}^2$
to the system (\ref{eqn1_1}) near $\boldsymbol{\phi}^{(\epsilon = 0)}(\boldsymbol{\theta}^{(\epsilon = 0)}) \in \mathbb{C}^2$
for every $\epsilon \in \mathcal{O}(0)$
such that the mapping $\epsilon \mapsto \boldsymbol{\phi}$  is smooth
and $\| \boldsymbol{\phi} - \boldsymbol{\phi}^{(\epsilon = 0)}(\boldsymbol{\theta}^{(\epsilon = 0)}) \| \leq C |\epsilon|$
for an $\epsilon$-independent positive constant $C$.
\end{corollary}

\begin{proof}
The proof is just an application of the two-step Lyapunov--Schmidt reduction method.
\end{proof}

In order to classify all possible solutions of the algebraic system (\ref{eqn1_1}) for small $\epsilon$,
according to the combined result of Lemmas \ref{lemma2_1} and \ref{lemma2_2},
we write $\boldsymbol{H}(\boldsymbol{\theta})$ and $\mathcal{N}(\boldsymbol{\theta})$ explicitly as:
\begin{eqnarray}
\label{eqn2_H}
\boldsymbol{H}(\boldsymbol{\theta}) = \left[
\begin{matrix} \sin(2 \theta_1) - \sin(\theta_2 - \theta_1) + \gamma_1 \\
\sin(2 \theta_2) - \sin(\theta_1 - \theta_2) + \gamma_2 \end{matrix} \right]
\end{eqnarray}
and
\begin{align}
\label{eqn2_N}
\mathcal{N}(\boldsymbol{\theta}) =
\left[
    \begin{matrix}
        2\cos(2 \theta_1) + \cos(\theta_2 - \theta_1) & -\cos(\theta_2 - \theta_1) \\
        -\cos(\theta_1 - \theta_2) & \cos(\theta_1 - \theta_2) + 2 \cos(2 \theta_2)
    \end{matrix}
\right].
\end{align}

Let us simplify the computations in the particular case $\gamma_1 = -\gamma_2 = \gamma$.
In this case, the system $\boldsymbol{H}(\boldsymbol{\theta}) = \boldsymbol{0}$ is equivalent to the system
\begin{eqnarray}
\label{eqn1_4}
\left\{ \begin{array}{l} \sin(2 \theta_1) + \sin(2 \theta_2) = 0, \\
\sin(2 \theta_1) - \sin(\theta_2 - \theta_1) + \gamma = 0. \end{array} \right.
\end{eqnarray}
The following list represents all families of solutions of the system (\ref{eqn1_4}),
which are uniquely continued to the solution of the system (\ref{eqn1_1})
for $\epsilon \neq 0$, according to the result of Corollary \ref{cor-2-1}.

\begin{itemize}
\item[(1-1)] Solving the first equation of system (\ref{eqn1_4}) with $2 \theta_2 = 2 \theta_1 + \pi$ and
the second equation with $\sin(2 \theta_1) = 1-\gamma$, we obtain a solution for $\gamma\in(0,2)$.
Two branches exist for $\theta_1 = \frac{1}{2}\arcsin(1-\gamma)$ and $\theta_1 =\frac{\pi}{2}-\frac{1}{2}\arcsin(1-\gamma)$,
which are denoted by (1-1-a) and (1-1-b), respectively. However, the branch (1-1-b)
is obtained from the branch (1-1-a) by using symmetries in Remarks \ref{remark1_1} and \ref{remark1_3}
as well as by flipping the configuration on the left panel of Figure \ref{vortex1} about the horizontal axis.
Therefore, it is sufficient to consider the branch (1-1-a) only. The Jacobian matrix in (\ref{eqn2_N}) is given by
\begin{equation}
\label{Jacobian-1}
2 \cos(2 \theta_1) \left[ \begin{matrix} 1 & 0 \\ 0 & -1 \end{matrix} \right]
\end{equation}
and it is invertible if $\cos(2 \theta_1) \neq 0$, that is, if $\gamma\neq 0, 2$.
In the limit $\gamma \to 0$, the solution $(\phi_1,\phi_2,\phi_3,\phi_4)$
along the branches (1-1-a) and (1-1-b) transforms to the limiting solution
$\left(e^{\frac{i \pi}{4}},e^{\frac{3\pi i}{4}},e^{\frac{5 \pi i}{4}},e^{\frac{7\pi i}{4}}\right)$,
which is the discrete vortex of charge one, according to the terminology in \cite{kev1}.
No vortex of the negative charge one exists for $\gamma \in (0,2)$.

\item[(1-2)] Solving the first equation of system (\ref{eqn1_4}) with $2 \theta_2 = 2 \theta_1 - \pi$ and
the second equation with $\sin(2 \theta_1) = -1-\gamma$, we obtain a solution for $\gamma\in(-2,0)$.
Two branches exist for $\theta_1 = -\frac{1}{2}\arcsin(1+\gamma)$ and $\theta_1 = -\frac{\pi}{2} + \frac{1}{2}\arcsin(1+\gamma)$,
which are denoted by (1-2-a)  and (1-2-b), respectively. Since the branches (1-1-a) and (1-1-b)
are connected to the branches (1-2-a) and (1-2-b) by Remark \ref{remark1_1},
it is again sufficient to limit our consideration by branch (1-1-a) for $\gamma \in (0,2)$.
In the limit $\gamma \to 0$, the solution $(\phi_1,\phi_2,\phi_3,\phi_4)$
along the branches (1-2-a) and (1-2-b) transforms to the limiting solution
$\left(e^{-\frac{i \pi}{4}},e^{-\frac{3\pi i}{4}},e^{-\frac{5 \pi i}{4}},e^{-\frac{7\pi i}{4}}\right)$
which is the discrete vortex of the negative charge one. No vortex of the positive charge one
exists for $\gamma \in (-2,0)$.

\item[(1-3)] Solving the first equation of system (\ref{eqn1_4}) with $2 \theta_2 = -2 \theta_1$ and
the second equation with $\sin(2 \theta_1) = -\frac{\gamma}{2}$, we obtain a solution for $\gamma\in(-2,2)$.
Two branches exist for $\theta_1 = -\frac{1}{2}\arcsin\left(\frac{\gamma}{2}\right)$ and
$\theta_1 =\frac{\pi}{2}+\frac{1}{2}\arcsin\left(\frac{\gamma}{2}\right)$,
which are denoted by (1-3-a) and (1-3-b), respectively.
The Jacobian matrix in (\ref{eqn2_N}) is given by
\begin{equation}
\label{Jacobian-2}
\cos(2 \theta_1) \left[ \begin{matrix} 3 & -1 \\ -1 & 3 \end{matrix} \right],
\end{equation}
which is invertible if $\cos(2 \theta_1) \neq 0$, that is, if $\gamma \neq \pm 2$.
In the limit $\gamma \to 0$, the solution $(\phi_1,\phi_2,\phi_3,\phi_4)$
along the branches (1-3-a) and (1-3-b) transforms to the limiting solutions
$(1,1,1,1)$ and $i (1,-1,1,-1)$, which correspond to discrete solitons, according to the terminology in \cite{kev1}.

\item[(1-4)] Solving the first equation of system (\ref{eqn1_4}) with $2 \theta_2 = -2 \theta_1 \pm 2 \pi$,
we obtain the constraint $\gamma = 0$ from the second equation.
Therefore, no solutions exist in this choice if $\gamma \neq 0$.
\end{itemize}

We conclude that the $\pt$-symmetry about the vertical line supports both vortex and soliton configurations.

\subsection{Symmetry about the center (S2)}

Under conditions $\gamma_1=-\gamma_3$, $\gamma_2=-\gamma_4$, $\phi_1 = \bar{\phi}_3$, and
$\phi_2 = \bar{\phi}_4$, the system (\ref{eqn0_1}) reduces to two algebraic equations:
\begin{eqnarray}
\label{eqn3_1}
\left\{ \begin{array}{l}
f_1 := (1-|\phi_1|^2)\phi_1 - \epsilon( \bar{\phi}_2 + \phi_2 - i\gamma_1 \phi_1 ) = 0, \\
f_2 := (1-|\phi_2|^2)\phi_2 - \epsilon( \bar{\phi}_1 + \phi_1 - i\gamma_2 \phi_2 ) = 0. \end{array} \right.
\end{eqnarray}
The system (\ref{eqn3_1}) is only slightly different from the system (\ref{eqn1_1}).
Therefore, Lemmas \ref{lemma2_1} and \ref{lemma2_2} hold for the system (\ref{eqn3_1})
and the question of persistence of vortex configurations symmetric about the center
can be solved with the two-step Lyapunov--Schmidt reduction algorithm. For explicit computations
of the persistence analysis, we obtain the explicit expressions for
$\boldsymbol{H}(\boldsymbol{\theta})$ and $\mathcal{N}(\boldsymbol{\theta})$ in Lemma \ref{lemma2_2}:
\begin{eqnarray}
\label{eqn3_H}
\boldsymbol{H}(\boldsymbol{\theta}) = \left[
\begin{matrix} \sin(\theta_2 + \theta_1) - \sin(\theta_2 - \theta_1) + \gamma_1 \\
\sin(\theta_1 + \theta_2) - \sin(\theta_1 - \theta_2) + \gamma_2 \end{matrix} \right]
\end{eqnarray}
and
\begin{align}
\label{eqn3_N}
\mathcal{N}(\boldsymbol{\theta}) =
\left[
    \begin{matrix}
        \cos(\theta_2 + \theta_1) + \cos(\theta_2 - \theta_1) & \cos(\theta_2 + \theta_1) -\cos(\theta_2 - \theta_1) \\
        \cos(\theta_1 + \theta_2) -\cos(\theta_1 - \theta_2) & \cos(\theta_1 + \theta_2) + \cos(\theta_1 - \theta_2)
    \end{matrix}
\right].
\end{align}

Let us now consider the $\pt$-symmetric network with $\gamma_1 = -\gamma_2 = \gamma$.
Therefore, we rewrite the system $\boldsymbol{H}(\boldsymbol{\theta}) = \boldsymbol{0}$
in the equivalent form:
\begin{eqnarray}
\label{eqn3_4}
\left\{ \begin{array}{l} \sin(\theta_1 + \theta_2) = 0, \\
\sin(\theta_1 - \theta_2) + \gamma = 0. \end{array} \right.
\end{eqnarray}
Note in passing that the system (\ref{eqn3_1}) admits the exact solution in the polar
form $\phi_{j} = r_j e^{i \theta_j}$, $j = 1,2$ with
$$
r_1 = r_2 = \sqrt{1 - \epsilon \left( \cos(\theta_1 + \theta_2) + \cos(\theta_1 - \theta_2) \right)},
$$
where $\theta_1$ and $\theta_2$ are given by the roots of the system (\ref{eqn3_4}). The
Lyapunov--Schmidt reduction algorithm in Lemmas \ref{lemma2_1} and \ref{lemma2_2}
guarantees that these exact solutions are unique in the neighborhood of the limiting solution (\ref{limit-vortex}).

The following list represents all families of solutions of the system (\ref{eqn3_4}),
which are uniquely continued to the solution of the system (\ref{eqn3_1})
for $\epsilon \neq 0$, according to the result of Corollary \ref{cor-2-1}.

\begin{itemize}
\item[(2-1)] $\theta_2 = -\theta_1$ and $\sin(2\theta_1)=-\gamma$
with two branches $\theta_1 = -\frac{1}{2} \arcsin(\gamma)$ and
$\theta_1 = \frac{\pi}{2} + \frac{1}{2} \arcsin(\gamma)$ labeled as (2-1-a) and (2-1-b).
The two branches exist for $\gamma \in (-1,1)$. The Jacobian matrix in
(\ref{eqn3_N}) is given by
\begin{equation}
\label{Jacobian-3}
2 \left[ \begin{matrix} \cos^2(\theta_1) & \sin^2(\theta_1) \\ \sin^2(\theta_1) & \cos^2(\theta_1) \end{matrix} \right]
\end{equation}
and it is invertible if $\cos(2 \theta_1) \neq 0$, that is, if $\gamma\neq \pm 1$.
In the limit $\gamma \to 0$, the solution $(\phi_1,\phi_2,\phi_3,\phi_4)$
along the branches (2-1-a) and (2-1-b) transforms to the limiting solutions $(1,1,1,1)$ and $i(1,-1,-1,1)$,
which correspond to discrete solitons.

\item[(2-2)] $\theta_2 = -\theta_1 \pm \pi$ and $\sin(2\theta_1) = \gamma$
with two branches $\theta_1 = \frac{1}{2} \arcsin(\gamma)$ and
$\theta_1 = \frac{\pi}{2} - \frac{1}{2} \arcsin(\gamma)$ labeled as (2-2-a) and (2-2-b).
These two branches also exist for $\gamma \in (-1,1)$.
The Jacobian matrix in (\ref{eqn3_N}) is given by
\begin{equation}
\label{Jacobian-4}
-2 \left[ \begin{matrix} \cos^2(\theta_1) & \sin^2(\theta_1) \\ \sin^2(\theta_1) & \cos^2(\theta_1) \end{matrix} \right],
\end{equation}
which is invertible if $\cos(2 \theta_1) \neq 0$, that is, if $\gamma\neq \pm 1$.
In the limit $\gamma \to 0$, the solution $(\phi_1,\phi_2,\phi_3,\phi_4)$ along the branches (2-2-a) and (2-2-b)
transforms to the limiting solutions $(1,-1,1,-1)$ and $i(1,1,-1,-1)$,
which again correspond to discrete solitons. The family (2-2) is related to the family (2-1) by Remark \ref{remark1_2}.
\end{itemize}

If we consider the $\pt$-symmetric network with $\gamma_1 = \gamma_2 = \gamma$, this case
is reduced to the network with $\gamma = -\gamma_2 = \gamma$ by flipping the second and fourth
sites about the center of symmetry in the configuration on the right panel of Figure \ref{vortex1}.
Therefore, we do not need to consider this case separately.

We conclude that the $\pt$-symmetry about the center supports only discrete soliton
configurations in the limit $\gamma \to 0$. No vortex configurations persist with respect
to $\epsilon$ if $\gamma \neq 0$.

\section{Existence of $\pt$-symmetric vortices in truncated lattice}
\label{section-full}

We shall now consider the $\pt$-symmetric dNLS equation (\ref{pt-dnls-equiv})
on the square lattice truncated symmetrically with suitable boundary conditions.

For the steady-state solutions $u_j(\tau) = \phi_j e^{-4 i \epsilon \tau}$, we obtain the stationary
$\pt$-symmetric dNLS equation in the form
\begin{eqnarray}
\label{eqn0_2}
(1-|\phi_j|^2)\phi_{j,k} - \epsilon ( \phi_{j+1,k} + \phi_{j-1,k} + \phi_{j,k+1}
+ \phi_{j, k-1} - i \gamma_{j,k} \phi_{j,k}) = 0, \quad (j,k) \in \mathbb{Z}^2.
\end{eqnarray}
In the limit of $\epsilon \to 0$, we are still looking for
the limiting configurations supported on four sites of the elementary cell:
\begin{equation}
\label{limiting-config-lattice}
\phi_{j,k}^{(\epsilon = 0)}(\theta_{j,k}) = e^{i \theta_{j,k}}, \quad (j,k) \in S := \left\{ (1,0); (1,1); (0,1); (0,0) \right\},
\end{equation}
where $\theta_{j,k} \in \mathbb{T}$ for $(j,k) \in S$, whereas $\phi_{j,k}^{(\epsilon = 0)} = 0$
for $(j,k) \in S^* := \mathbb{Z}^2 \backslash S$.

Computations in Section \ref{section-four-sites} remain valid on the unbounded square lattice,
because the results of Lemmas \ref{lemma2_1} and \ref{lemma2_2} are obtained on the set $S$ in the first order
in $\epsilon$, where no contributions come from the empty sites in the set $S^*$.

If the square lattice is truncated, then the truncated square lattice must satisfy
the following requirements for persistence of the $\pt$-symmetric configurations:
\begin{itemize}
\item the elementary cell $S$ must be central in the symmetric extension of the lattice,
\item the distribution of gains and losses in $\{ \gamma_{j,k} \}$
must be anti-symmetric with respect to the selected symmetry on Figure \ref{vortex1},
\item the boundary conditions must be consistent with the $\pt$-symmetry constraints on $\{ \phi_{j,k} \}$.
\end{itemize}
The periodic boundary conditions may not be consistent with the $\pt$-symmetry constraints
because of the jump in the complex phases. On the other hand,
Dirichlet conditions at the fixed ends are consistent with the $\pt$-symmetry constraints.

\begin{figure}[htbp]
\begin{tabular}{cc}
\includegraphics[width=7cm]{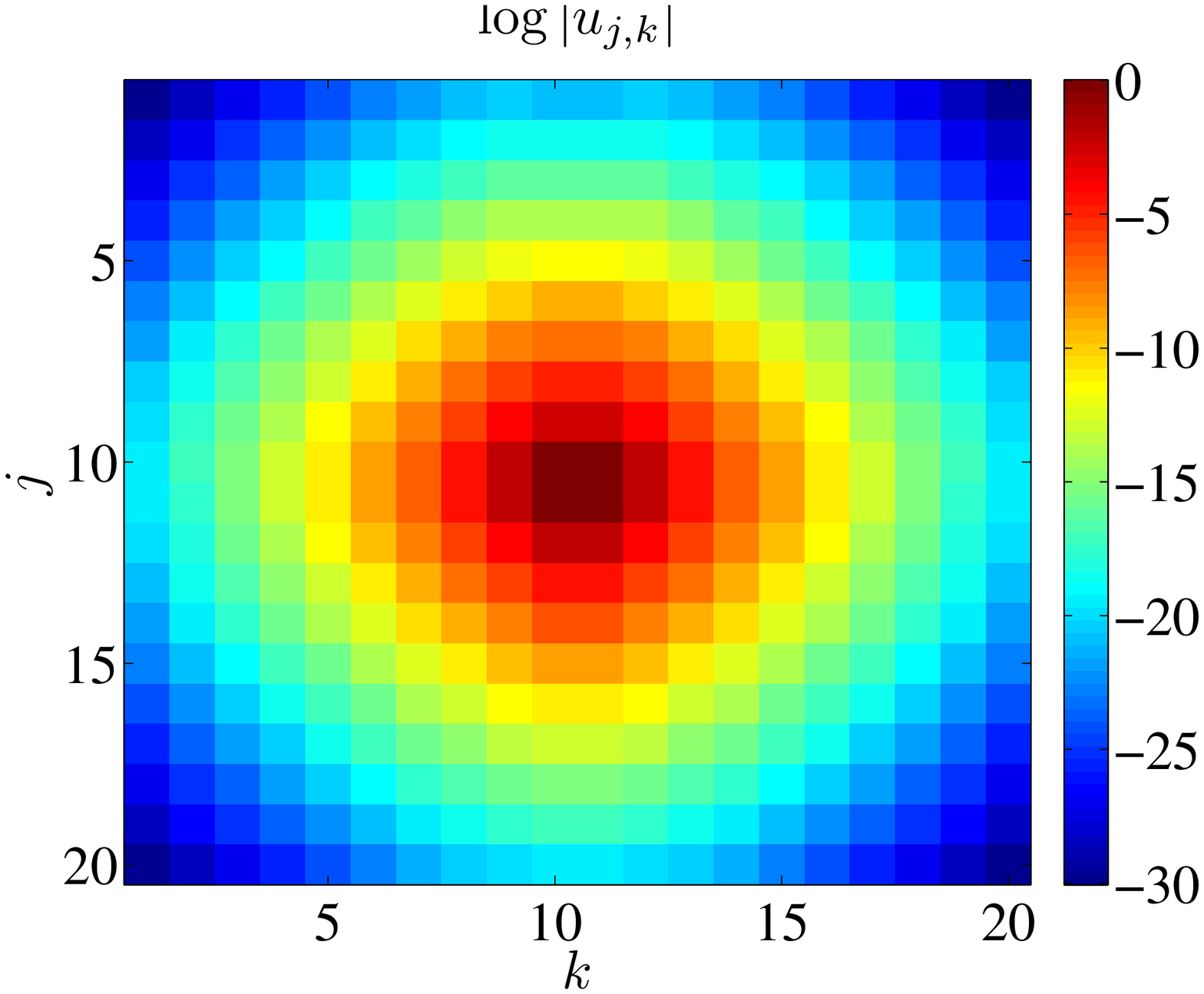}
\includegraphics[width=7cm]{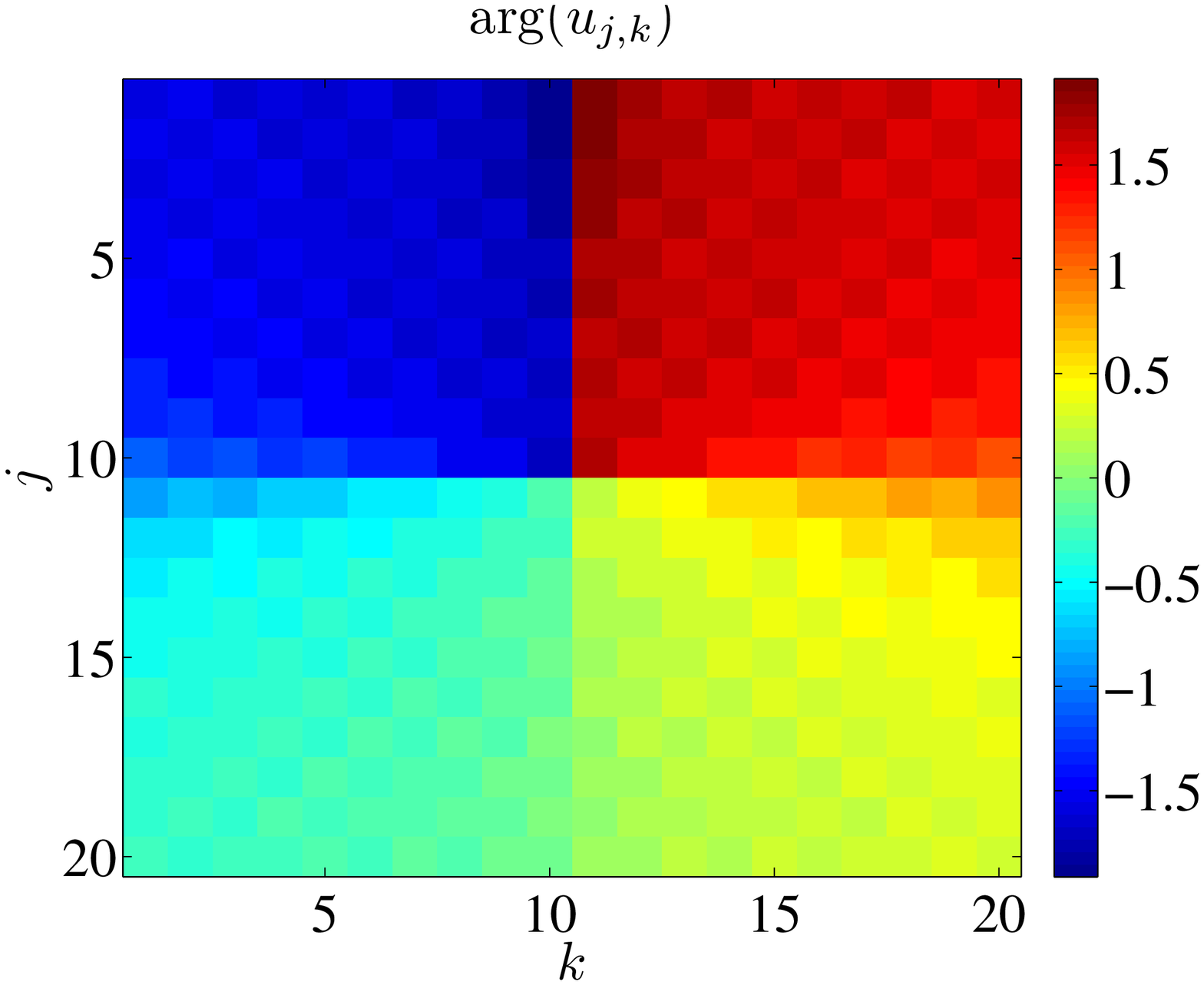}\\
\includegraphics[width=7cm]{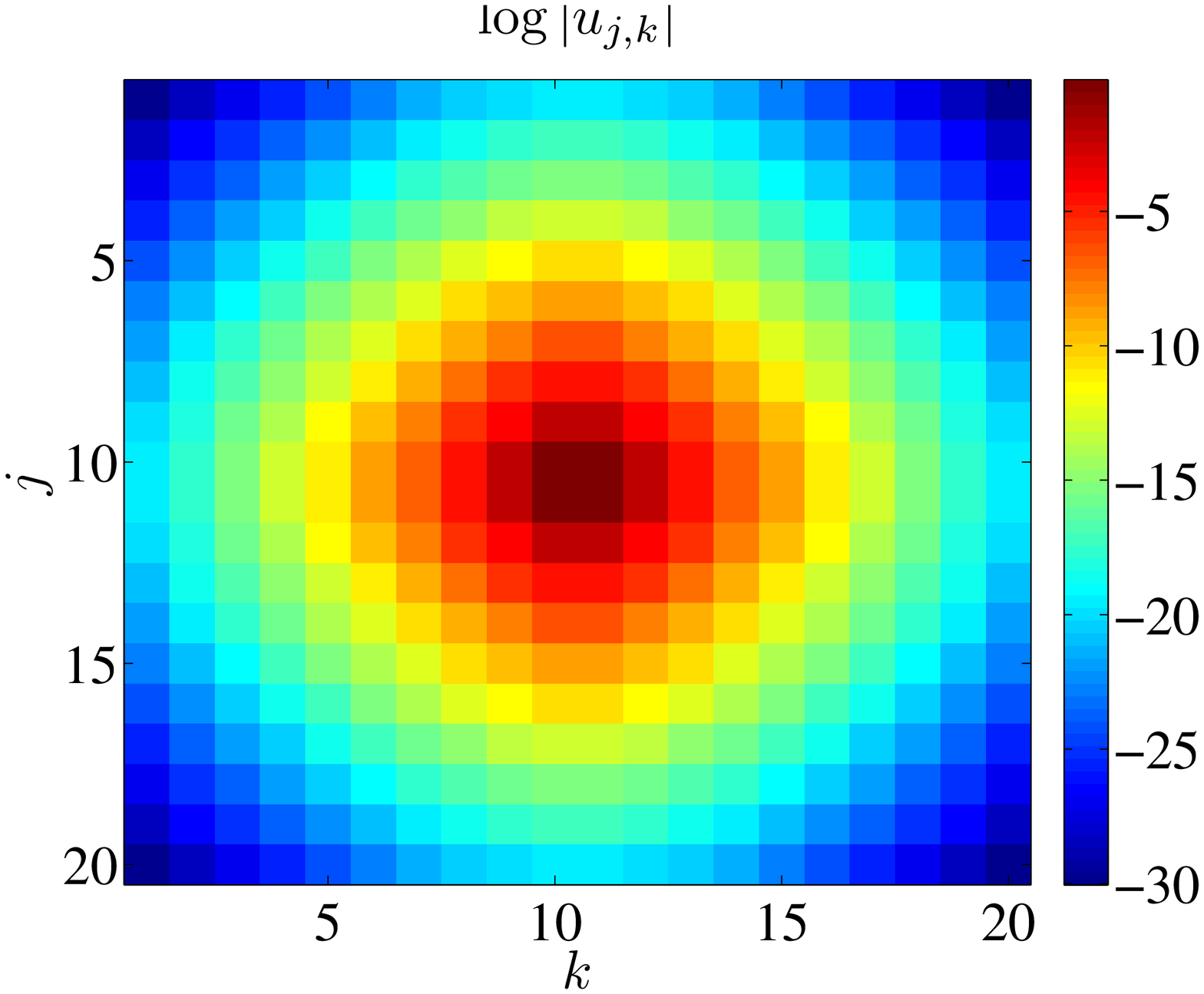}
\includegraphics[width=7cm]{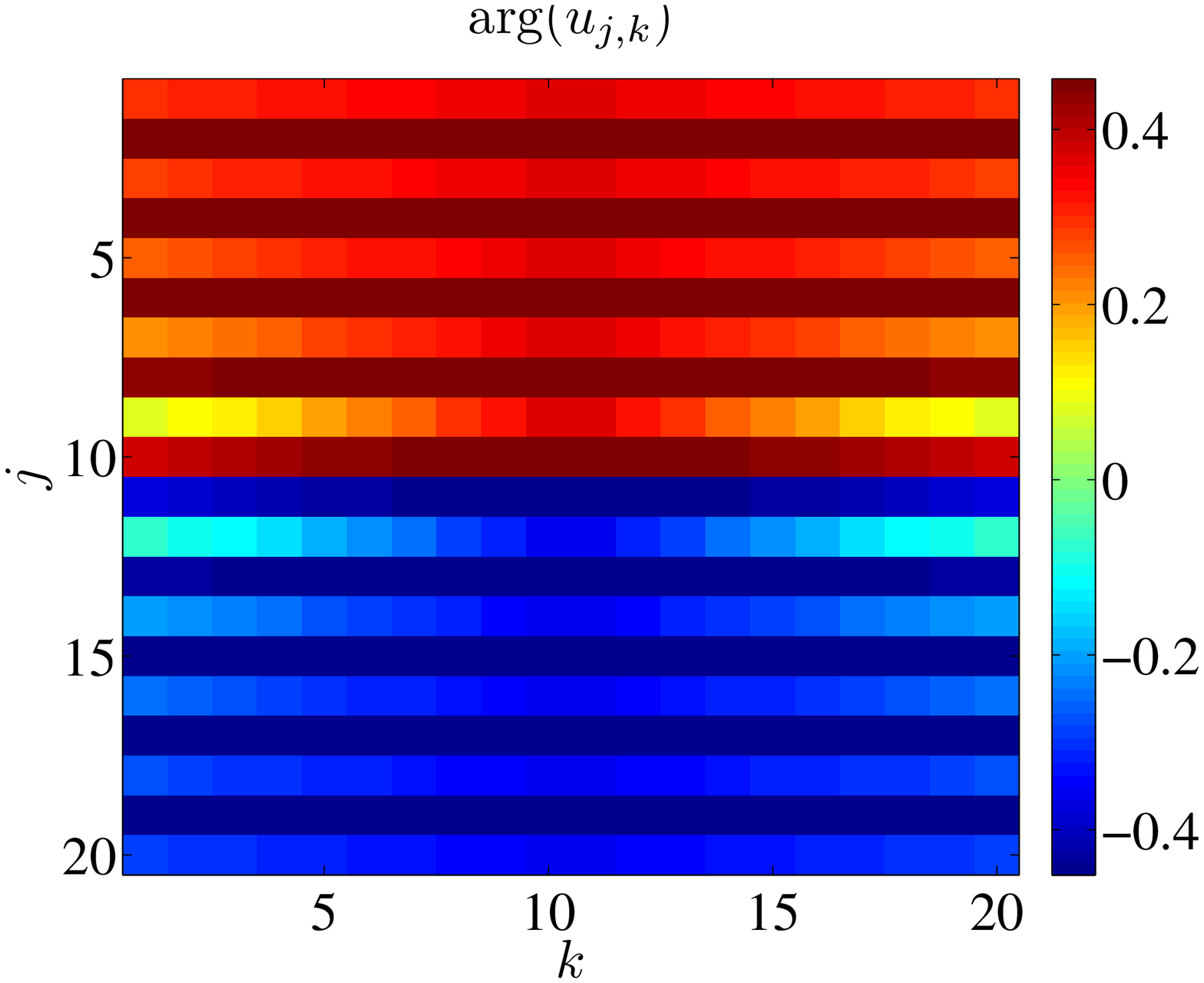}
\end{tabular}
\caption{The top panels show a continuation of the branch (1-1-a) on the $20$-by-$20$ square lattice
with $\gamma_1=-\gamma_2=0.7$ and $\epsilon=0.1$. The bottom panels show a continuation of the branch
(2-1-a) with $\gamma_1=-\gamma_2=0.8$ and $\epsilon=0.1$. The
left panels show the logarithm of the solution's modulus, while
the right panels show the corresponding phase.}
\label{fig3_1}
\end{figure}

In Figure \ref{fig3_1}, we show continuations of the $\pt$-symmetric solutions
from the branches (1-1-a) and (2-1-a) obtained on the elementary cell $S$ in Section \ref{section-four-sites}.
The relevant configurations are now computed on the $20$-by-$20$ square lattice
truncated symmetrically with zero boundary conditions. The requirements listed above are satisfied
for the truncated square lattice. The left panels of Figure \ref{fig3_1} illustrate the
logarithm of the modulus, while the right panels show the corresponding phase.

At $\epsilon=0$, the phases of the solution on $S$ in (1-1-a) are
$\{ \theta_1, \theta_1+\frac{\pi}{2}, -\theta_1-\frac{\pi}{2}, -\theta_1 \}$,
where $\theta_1 =\frac{1}{2}\arcsin(1-\gamma)$. Therefore, the configuration
represents the continuation over $\gamma$ of the
discrete vortex of charge one, which corresponds to $\theta_1 = \pi/4$ at $\gamma = 0$.
This interpretation is confirmed by the
surface plot for the argument of complex amplitudes, which shows a $2\pi$-change over a discrete contour
surrounding the vortex location.
On the other hand, the solution on $S$ in (2-1-a) has phases
$\{ \theta_1, -\theta_1, -\theta_1, \theta_1 \}$, where $\theta_1 =-\frac{1}{2}\arcsin(\gamma)$.
Therefore, this configuration represents the
continuation over $\gamma$ of the discrete soliton,
which corresponds to $\theta_1 = 0$ at $\gamma = 0$. This is again confirmed by
the surface plot for the argument of complex amplitudes.

Since $\epsilon = 0.1$ is small, we can observe  in Figure \ref{fig3_1} that the solutions
are still close to the limiting solutions of $\epsilon=0$ and the amplitudes
are large chiefly
at the four central sites of the set $S$. However, amplitudes of the sites in the set $S^*$
are nonzero for $\epsilon = 0.1$ but still small (at most $\mathcal{O}(\epsilon)$ on the sites
of $S^*$ adjacent to the sites of $S$). At the same time, the
phases of the amplitudes on the sites of $S$ do not change much in parameter $\epsilon$
and still feature a clearly discernible discrete vortex with charge one in the continuation of the branch (1-1-a)
and a discrete soliton in the continuation of the branch (2-1-a).

\section{Stability of $\pt$-symmetric configurations in the cell}
\label{section-stability}

We address the $\pt$-symmetric configurations in the elementary cell consisting of the four sites.
Persistence of these configurations in small parameter $\epsilon$ is obtained
in Section \ref{section-four-sites}. In what follows, we consider stability of the $\pt$-symmetric configurations
in the sense of spectral stability.

Let $\boldsymbol{\phi} := \{ \phi_j \}_{1 \leq j \leq 4}$ be a stationary solution of the system (\ref{eqn0_1}).
If it is $\pt$-symmetric, there exists a $4$-by-$4$ matrix $P$ such that $\bar{\boldsymbol{\phi}} =
P \boldsymbol{\phi}$. For the two $\pt$-symmetries considered in Section \ref{section-four-sites},
we list the matrix $P$:
\begin{equation}
\label{symmetry-matrix}
\mbox{\rm (S1)} \quad P =
\left(
    \begin{array}{cccc}
        0 & 0 & 0 & 1 \\
        0 & 0 & 1 & 0 \\
        0 & 1 & 0 & 0 \\
        1 & 0 & 0 & 0
    \end{array}
\right), \quad \quad \quad
\mbox{\rm (S2)} \quad P =
\left(
    \begin{array}{cccc}
        0 & 0 & 1 & 0 \\
        0 & 0 & 0 & 1 \\
        1 & 0 & 0 & 0 \\
        0 & 1 & 0 & 0
    \end{array}
\right),
\end{equation}
Adding a perturbation to the stationary $\pt$-symmetric solution, we consider solutions
of the $\pt$-symmetric dNLS equation (\ref{pt-dnls-equiv}) in the form
$$
u_j(t) = \left( \phi_j + \delta \left[ e^{\lambda \tau} v_j + e^{\bar{\lambda} \tau}\overbar{w_j} \right] \right) e^{-2 i \epsilon \tau},
$$
where $\delta$ is the perturbation amplitude, $\lambda \in \mathbb{C}$ is the spectral parameter,
and $(\boldsymbol{v},\boldsymbol{w}) := \{ (v_j,w_j) \}_{1 \leq j \leq 4}$ represents an eigenvector of the spectral stability problem.
After the linearization of the $\pt$-symmetric dNLS equation (\ref{pt-dnls-equiv}), we obtain the spectral
stability problem in the form
\begin{eqnarray}
\label{eqn4_1}
\left\{ \begin{array}{l}
i \lambda v_j = (1+i\epsilon\gamma_j-2|\phi_j|^2)v_j -(\phi_j)^2 w_j -\epsilon(v_{j-1}+v_{j+1}), \\
- i\lambda w_j = -\overbar{\phi}_j^2 v_j + (1-i\epsilon\gamma_j-2|\phi_j|^2) w_j -\epsilon(w_{j-1}+w_{j+1}). \end{array} \right.
\end{eqnarray}
The eigenvalue problem (\ref{eqn4_1}) can be written in the matrix form
\begin{eqnarray}
\label{ev_matrix2}
i \lambda \sigma \boldsymbol{\xi} = \left( \mathcal{H} + i \epsilon \mathcal{G} \right) \boldsymbol{\xi},
\end{eqnarray}
where $\boldsymbol{\xi}$ consists of blocks of $(v_j,w_j)^T$, $\sigma$ consists of blocks of Pauli matrices
$\sigma_3 = {\rm diag}(1,-1)$, $\mathcal{G}$ consists of blocks of $\gamma_j \sigma_3$, and $\mathcal{H}$
is the Hermitian matrix consisting of the blocks of
\begin{align}
\label{eqn0_H}
\mathcal{H}_j=
\left(
    \begin{array}{cc}
        1-2|\phi_j|^2 & -(\phi_j)^2\\
        -(\overbar{\phi_j})^2 & 1-2|\phi_j|^2
    \end{array}
\right)
-\epsilon(s_{+1}+s_{-1})
\left(
    \begin{array}{cc}
        1 & 0\\
        0 & 1
    \end{array}
\right),
\end{align}
where $s_{j}$ stands for the shift operator such that $s_j \phi_k = \phi_{j+k}$.

\begin{remark}
\label{remark3_1}
If $\lambda$ is an eigenvalue of the spectral problem (\ref{eqn4_1}) with the eigenvector
$({\bf v},{\bf w})$, then $\bar{\lambda}$ is another eigenvalue
of the same problem (\ref{eqn4_1}) with the eigenvector $(\bar{\bf w},\bar{\bf v})$.
Therefore, eigenvalues $\lambda$ are symmetric about the real axis.
\end{remark}

\begin{remark}
\label{remark3_2}
Assume that $\boldsymbol{\phi}$ is $\pt$-symmetric, so that $\bar{\boldsymbol{\phi}} =
P \boldsymbol{\phi}$ for $P$ given by (\ref{symmetry-matrix}).
If $\lambda$ is an eigenvalue of the spectral problem (\ref{eqn4_1}) with the eigenvector
$({\bf v},{\bf w})$, then $-\bar{\lambda}$ is another eigenvalue
of the same problem (\ref{eqn4_1}) with the eigenvector $(P \bar{\bf v},P \bar{\bf w})$.
Therefore, eigenvalues $\lambda$ are symmetric about the imaginary axis.
\end{remark}

In order to study stability of the stationary solutions in the limit of small $\epsilon$,
we adopt the stability results obtained in \cite{kev1}. Along this way, it is easier
to work with a stationary solution $\boldsymbol{\phi}$ without using the property of
$\pt$-symmetry. Nevertheless,
it is true that $\phi_3$ and $\phi_4$ are expressed from $\phi_1$ and $\phi_2$ by using
the matrix $P$ given by (\ref{symmetry-matrix}). With the account of the relevant symmetry,
Corollary \ref{cor-2-1} implies that the stationary solution can be expressed in the form
\begin{equation}
\label{eqn0_2_1}
\phi_j = e^{i \theta_j^{(0)}} \left[ 1 + \epsilon r_j^{(1)}  +
i \epsilon \left( \theta_j^{(1)} - \theta_j^{(0)} \right)  +
\mathcal{O}(\epsilon^2) \right], \quad 1 \leq j \leq 4,
\end{equation}
where $\{ \theta_j^{(0)} \}_{1 \leq j \leq 4}$ are determined from simple roots of
the vector function $\boldsymbol{H}$, $\{ \theta_j^{(1)} \}_{1 \leq j \leq 4}$ are found
from persistence analysis in Lemma \ref{lemma2_2}, and $\{ r_j^{(1)} \}_{1 \leq j \leq 4}$
are found from persistence analysis in Lemma \ref{lemma2_1}. After elementary computations,
we obtain the explicit expression
\begin{equation}
\label{eqn0_2_2}
r_j^{(1)} = -\frac{1}{2} \left[ \cos(\theta_j^{(0)} - \theta_{j-1}^{(0)}) +
\cos(\theta_j^{(0)} - \theta_{j+1}^{(0)}) \right],
\end{equation}
where the cyclic boundary conditions for $\{ \theta_j^{(0)}\}_{1 \leq j \leq 4}$
are assumed. For convenience of our presentation, we drop the superscripts in
writing $\theta_j^{(0)}$.

Let $\mathcal{M}$ be a $4$-by-$4$ matrix satisfying
\begin{eqnarray}
\label{eqn0_M}
\mathcal{M}_{j,k}=
\left\{ \begin{array}{cr}
\cos(\theta_j -\theta_{j-1})+\cos(\theta_j-\theta_{j+1}), &k=j \\
-\cos(\theta_j -\theta_{j-1}), &k=j-1 \\
-\cos(\theta_j - \theta_{j+1}), &k=j+1 \\
0 & {\rm otherwise}. \end{array} \right.
\end{eqnarray}
Due to the gauge invariance of the original dNLS equation (\ref{pt-dnls-equiv}),
$\mathcal{M}$ always has a zero eigenvalue with eigenvector $(1,1,1,1)^T$.
The other three eigenvalues of $\mathcal{M}$ may be nonzero. As is established
in \cite{kev1}, the nonzero eigenvalues of the matrix $\mathcal{M}$ are related to
small eigenvalues of the matrix operator $\mathcal{H}$ of the order of $\mathcal{O}(\epsilon)$.
In order to render
the stability analysis herein self-contained, we review the statement and the proof of this result.

\begin{lemma}
\label{lemma1_1}
Let $\mu_j$ be a nonzero eigenvalue of the matrix $\mathcal{M}$.
Then, for sufficiently small $\epsilon \in \mathcal{O}(0)$,
the matrix operator $\mathcal{H}$ has a small nonzero eigenvalue $\nu_j$ such that
\begin{eqnarray}
\label{ev_lemma1}
\nu_j = \mu_j \epsilon + \mathcal{O}(\epsilon^2).
\end{eqnarray}
\end{lemma}

\begin{proof}
We consider the expansion $\mathcal{H}=\mathcal{H}^{(0)}+\epsilon\mathcal{H}^{(1)}+O(\epsilon^2)$,
where $\mathcal{H}^{(0)}$ consists of the blocks
\begin{align}
\label{eqn0_H0}
\mathcal{H}_j^{(0)}=
\left(
    \begin{array}{cc}
        -1 & -e^{2i\theta_j}\\
        -e^{-2i\theta_j} & -1
    \end{array}
\right)
\end{align}
Each block has a one-dimensional kernel spanned by the vector $( e^{i\theta_j}, -e^{-i\theta_j} )$. Let
$\boldsymbol{e_j}$ be the corresponding eigenvector of $\mathcal{H}^{(0)}$ for the zero eigenvalue.
Therefore, we have
$$
\ker(\mathcal{H}^{(0)}) = {\rm span}\{\boldsymbol{e_j}\}_{1\leq j\leq 4}.
$$

By regular perturbation theory, we are looking for the small eigenvalue $\nu_j$ and eigenvector $\boldsymbol{\eta}$
of the Hermitian matrix operator $\mathcal{H}$ for small $\epsilon \in \mathcal{O}(0)$ in the form
$$
\nu_j = \epsilon \nu_j^{(1)} + \mathcal{O}(\epsilon^2), \quad
\boldsymbol{\eta} = \boldsymbol{\eta}^{(0)} + \epsilon\boldsymbol{\eta}^{(1)} + \mathcal{O}(\epsilon^2),
$$
where $\boldsymbol{\eta}^{(0)}=\sum_{j=1}^4 c_j \boldsymbol{e_j}$ and $\{ c_j \}_{1 \leq j \leq 4}$ are to be determined.
At the first order of $\mathcal{O}(\epsilon)$, we obtain the linear inhomogeneous system
\begin{eqnarray}
\label{eqn_lemma1_1_d1}
\mathcal{H}^{(0)}\boldsymbol{\eta}^{(1)} + \mathcal{H}^{(1)}\boldsymbol{\eta}^{(0)}  =
\nu_j^{(1)}\boldsymbol{\eta}^{(0)},
\end{eqnarray}
where $\mathcal{H}^{(1)}$ consists of the blocks
\begin{eqnarray}
\label{eqn0_H1}
\mathcal{H}_j^{(1)} =
-2 r_j^{(1)} \left(
    \begin{array}{cc}
        2   & e^{2 i \theta_j} \\
        e^{-2 i \theta_j} & 2
    \end{array}
\right)
-(s_{+1}+s_{-1})
\left(
    \begin{array}{cc}
        1 & 0\\
        0 & 1
    \end{array}
\right).
\end{eqnarray}
where the expansion (\ref{eqn0_2_1}) has been used. Projection of the linear inhomogeneous equation
(\ref{eqn_lemma1_1_d1}) to $\ker(\mathcal{H}^{(0)})$ gives the $4$-by-$4$ matrix eigenvalue problem
\begin{align}
\label{eqn_lemma1_1_B}
\mathcal{M} \boldsymbol{c}=\nu_j^{(1)}\boldsymbol{c},
\end{align}
where $\mathcal{M}_{j,k} = \frac{1}{2} \langle \boldsymbol{e_j}, \mathcal{H}^{(1)}\boldsymbol{e_k} \rangle$
is found to coincide with the one given by (\ref{eqn0_M}) thanks to the explicit expressions (\ref{eqn0_2_2}).
\end{proof}

We will now prove that the small nonzero eigenvalues of $\mathcal{H}$ for small nonzero
$\epsilon$ determine the small eigenvalues in the spectral stability problem (\ref{ev_matrix2}).
The following lemma follows the approach of \cite{kev1} but incorporates the additional
term $i \epsilon \mathcal{G}$ due to the $\pt$-symmetric gain and loss terms.

\begin{lemma}
\label{lemma1_2}
Let $\mu_j$ be a nonzero eigenvalue of the matrix $\mathcal{M}$.
Then, for sufficiently small $\epsilon \in \mathcal{O}(0)$,
the spectral problem (\ref{ev_matrix2}) has a small nonzero eigenvalue $\lambda_j$ such that
\begin{eqnarray}
\label{ev_lemma2}
\lambda_j^2 = 2 \mu_j \epsilon + \mathcal{O}(\epsilon^2).
\end{eqnarray}
\end{lemma}

\begin{proof}
We recall that $\ker(\mathcal{H}^{(0)}) = {\rm span}\{\boldsymbol{e_j}\}_{1\leq j\leq 4}$.
Let $\hat{\boldsymbol{e_j}} = \sigma \boldsymbol{e_j}$. Then, $\mathcal{H}^{(0)} \hat{\boldsymbol{e_j}} = -2\hat{\boldsymbol{e_j}}$.
Since the operator $\sigma \mathcal{H}^{(0)}$ is not self-adjoint, the zero eigenvalue of $\mathcal{H}^{(0)}$ of geometric multiplicity
$4$ may become a defective zero eigenvalue of $\mathcal{H}^{(0)}$ of higher algebraic multiplicity.
As is well-known \cite{kev1}, the zero eigenvalue of $\mathcal{H}^{(0)}$ has algebraic multiplicity $8$ with
$$
\ker(\sigma \mathcal{H}^{(0)}) = {\rm span}\{\boldsymbol{e_j}\}_{1\leq j\leq 4} \quad \mbox{\rm and} \quad
\ker((\sigma \mathcal{H}^{(0)})^2) = {\rm span}\{\boldsymbol{e_j}, \hat{\boldsymbol{e_j}}\}_{1\leq j\leq 4}.
$$

By the regular perturbation theory for two-dimensional Jordan blocks,
we are looking for the small eigenvalue $\lambda_j$ and eigenvector $\boldsymbol{\xi}$
of the spectral problem (\ref{ev_matrix2}) for small $\epsilon \in \mathcal{O}(0)$ in the form
$$
\lambda_j = \epsilon^{1/2} \lambda_j^{(1)} + \epsilon \lambda_j^{(2)} + \mathcal{O}(\epsilon^{3/2}), \quad
\boldsymbol{\xi} = \boldsymbol{\xi}^{(0)} + \epsilon^{1/2} \boldsymbol{\xi}^{(1)}
+ \epsilon \boldsymbol{\xi}^{(2)} + \mathcal{O}(\epsilon^{3/2}),
$$
where $\boldsymbol{\xi}^{(0)}=\sum_{j=1}^4 c_j \boldsymbol{e_j}$ and $\{ c_j \}_{1 \leq j \leq 4}$ are to be determined.
At the first order of $\mathcal{O}(\epsilon^{1/2})$, we obtain the linear inhomogeneous system
\begin{eqnarray}
\label{eqn_lemma1_2_d1}
\mathcal{H}^{(0)} \boldsymbol{\xi}^{(1)} = i\lambda_j^{(1)} \sigma \boldsymbol{\xi}^{(0)}.
\end{eqnarray}
Since $\sigma\mathcal{H}^{(0)}\boldsymbol{\hat{e}_j}=-2\boldsymbol{e_j}$,
we obtain the explicit solution of the linear inhomogeneous equation (\ref{eqn_lemma1_2_d1}) in the form
$$
\boldsymbol{\xi}^{(1)}=-\frac{i\lambda_j^{(1)}}{2}\sum_{j=1}^{4} c_j \boldsymbol{\hat{e}_j}.
$$
At the second order of $\mathcal{O}(\epsilon)$, we obtain the linear inhomogeneous system
\begin{eqnarray}
\label{eqn_lemma1_2_d2}
\mathcal{H}^{(0)}\boldsymbol{\xi}^{(2)} + \mathcal{H}^{(1)} \boldsymbol{\xi}^{(0)}
+ i \mathcal{G} \boldsymbol{\xi}^{(0)} = i\lambda_j^{(1)}\sigma\boldsymbol{\xi}^{(1)} + i\lambda_j^{(2)}\sigma\boldsymbol{\xi}^{(0)}.
\end{eqnarray}
Projection of the linear inhomogeneous equation
(\ref{eqn_lemma1_2_d2}) to $\ker(\mathcal{H}^{(0)})$ gives the $4$-by-$4$ matrix eigenvalue problem
\begin{align}
\label{eqn_lemma1_2_B2}
\mathcal{M} \boldsymbol{c}=\frac{1}{2} (\lambda_j^{(1)})^2 \boldsymbol{c},
\end{align}
since $\langle \boldsymbol{e_j}, \mathcal{G} \boldsymbol{e_k} \rangle = 0$ for every $j,k$.
Thus, the additional term $i \epsilon \mathcal{G}$ due to the $\pt$-symmetric gain and loss terms
does not contribute to the leading order of the nonzero eigenvalues $\lambda_j$,
which split according to the asymptotic expansion (\ref{ev_lemma2}). Note that the
relevant eigenvalues still depend on the gain--loss parameter
$\gamma$, due to the dependence of the parameters $\{ \theta_j \}_{1 \leq j \leq 4}$ on $\gamma$.
\end{proof}

If the stationary solution $\{ \phi_j \}_{1 \leq j \leq 4}$ satisfies the $\pt$-symmetry,
that is, $\bar{\boldsymbol{\phi}} = P \boldsymbol{\phi}$ with $P$ given by (\ref{symmetry-matrix}),
then the $4$-by-$4$ matrix $\mathcal{M}$ given by (\ref{eqn0_M}) has additional symmetry and can be folded
into two $2$-by-$2$ matrices. One of these two matrices must have nonzero eigenvalues for the
$\pt$-symmetric configurations because the configuration persists with respect to the small parameter $\epsilon$
by Corollary \ref{cor-2-1}. The other matrix must have a zero eigenvalue due to the gauge invariance
of the system of stationary equations \eqref{eqn0_1}. Since the persistence analysis depends
on the $\pt$-symmetry and different solution
branches have been identified in each case, we continue separately for the two kinds of the $\pt$-symmetry on the elementary cell
studied in Section \ref{section-four-sites}.

\subsection{Symmetry about the vertical line (S1)}

Under conditions $\gamma_1 = -\gamma_4$ and $\gamma_2 = -\gamma_3$, we consider the $\pt$-symmetric
configuration in the form $\phi_1=\bar{\phi}_4$ and $\phi_2=\bar{\phi}_3$.
Thus, we have $\theta_1 = -\theta_4$ and $\theta_2 = -\theta_3$, after which
the $4$-by-$4$ matrix $\mathcal{M}$ given by (\ref{eqn0_M}) can be written in the
explicit form:
\begin{equation*}
\mathcal{M}=
\left(
    \begin{array}{cccc}
        a + b & -b & 0 & -a \\
        -b & b + c & -c & 0 \\
        0 & -c & b + c & -b \\
        -a & 0 & -b & a + b
    \end{array}
\right),
\end{equation*}
where $a = \cos(2\theta_1)$, $b = \cos(\theta_1 -\theta_2)$, and $c = \cos(2\theta_2)$.
Using the transformation matrix
\begin{equation*}
T = \left(
    \begin{array}{cccc}
        1 & 0 & 0 & 1 \\
        0 & 1 & 1 & 0 \\
        0 & 1 & -1 & 0 \\
        1 & 0 & 0 & -1
    \end{array}
\right),
\end{equation*}
which generalizes the matrix $P$ given by (\ref{symmetry-matrix}) for (S1), we obtain the block-diagonalized
form of the matrix $M$ after a similarity transformation:
\begin{equation*}
T^{-1} M T = \left(
    \begin{array}{cccc}
        b & -b & 0 & 0 \\
        -b & b & 0 & 0 \\
        0 & 0 & b+2c & -b \\
        0 & 0 & -b & 2a + b
    \end{array}
\right).
\end{equation*}
Note that the first block is given by a singular matrix, whereas the second block coincides
with the Jacobian matrix $\mathcal{N}(\boldsymbol{\theta})$ given by (\ref{eqn2_N}).

The following list summarizes the stability features of the irreducible branches (1-1-a), (1-3-a), and (1-3-b)
among solutions of the system (\ref{eqn1_4}), which corresponds to the particular case $\gamma_1 = -\gamma_2 = \gamma$.

\begin{itemize}
\item[(1-1-a)] For the solution with $2 \theta_2 = 2 \theta_1 + \pi$ and
$\sin(2 \theta_1) = 1-\gamma$, we obtain $a = -c = \cos(2 \theta_1)$ and $b = 0$.
Therefore, the matrix $\mathcal{M}$ has a double zero eigenvalue
with eigenvectors $(1,0,0,1)^T, (0,1,1,0)^T$ and a pair of simple nonzero eigenvalues
$\mu_{\pm} = \pm 2 \cos(2\theta_1)$ with eigenvector $(-1,0,0,1)^T$ and $(0,-1,1,0)^T$.

By Lemma~\ref{lemma1_2}, the spectral stability problem (\ref{ev_matrix2}) has two pairs of small
nonzero eigenvalues $\lambda$: one pair of $\lambda$ is purely real near $\pm \sqrt{|4\epsilon \cos(2\theta_1)|}$
and another pair of $\lambda$ is purely imaginary near $\pm i\sqrt{|4\epsilon \cos(2\theta_1)|}$ as $\epsilon\to 0$.
Therefore, the branch (1-1-a) corresponding to the vortex configurations is spectrally unstable.
This instability disappears in the limit of $\gamma \rightarrow 0$, as
$\theta_1 \rightarrow \pi/4$ in this limit, and the dependence
of the relevant eigenvalue emerges at a higher order
in $\epsilon$, with the eigenvalue being imaginary~\cite{kev1}.

Figure \ref{fig1_1} shows comparisons between the eigenvalues approximated using the
first-order reductions in Lemma \ref{lemma1_2} and those computed numerically for the branch (1-1-a)
with $\epsilon > 0$. Note that one of the pair of real eigenvalues (second thickest solid blue line) appears beyond
the first-order reduction of Lemma \ref{lemma1_2} (see \cite{kev1} for further details).

\begin{figure}[!htbp]
\begin{tabular}{cc}
\includegraphics[width=7cm]{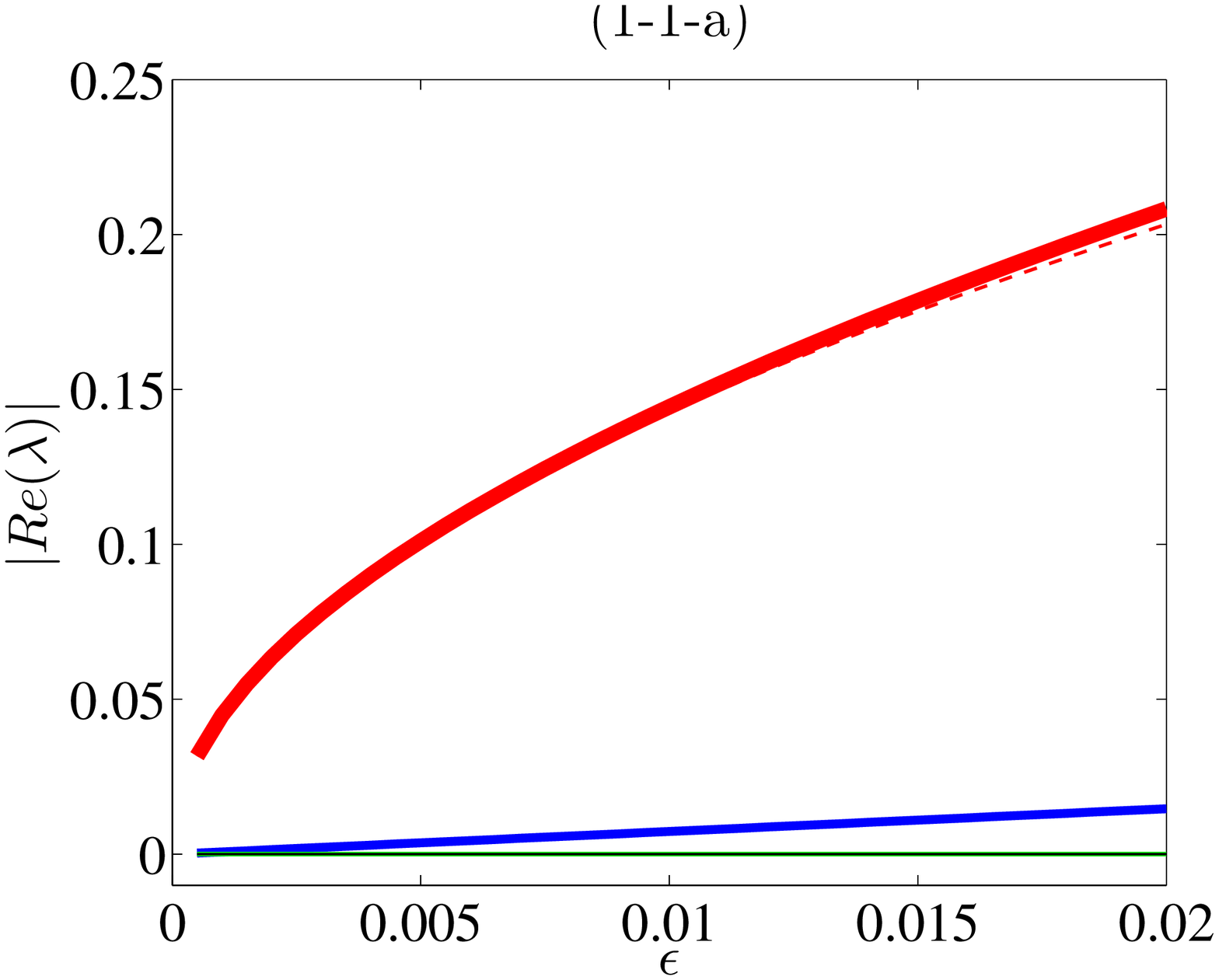}
\includegraphics[width=7cm]{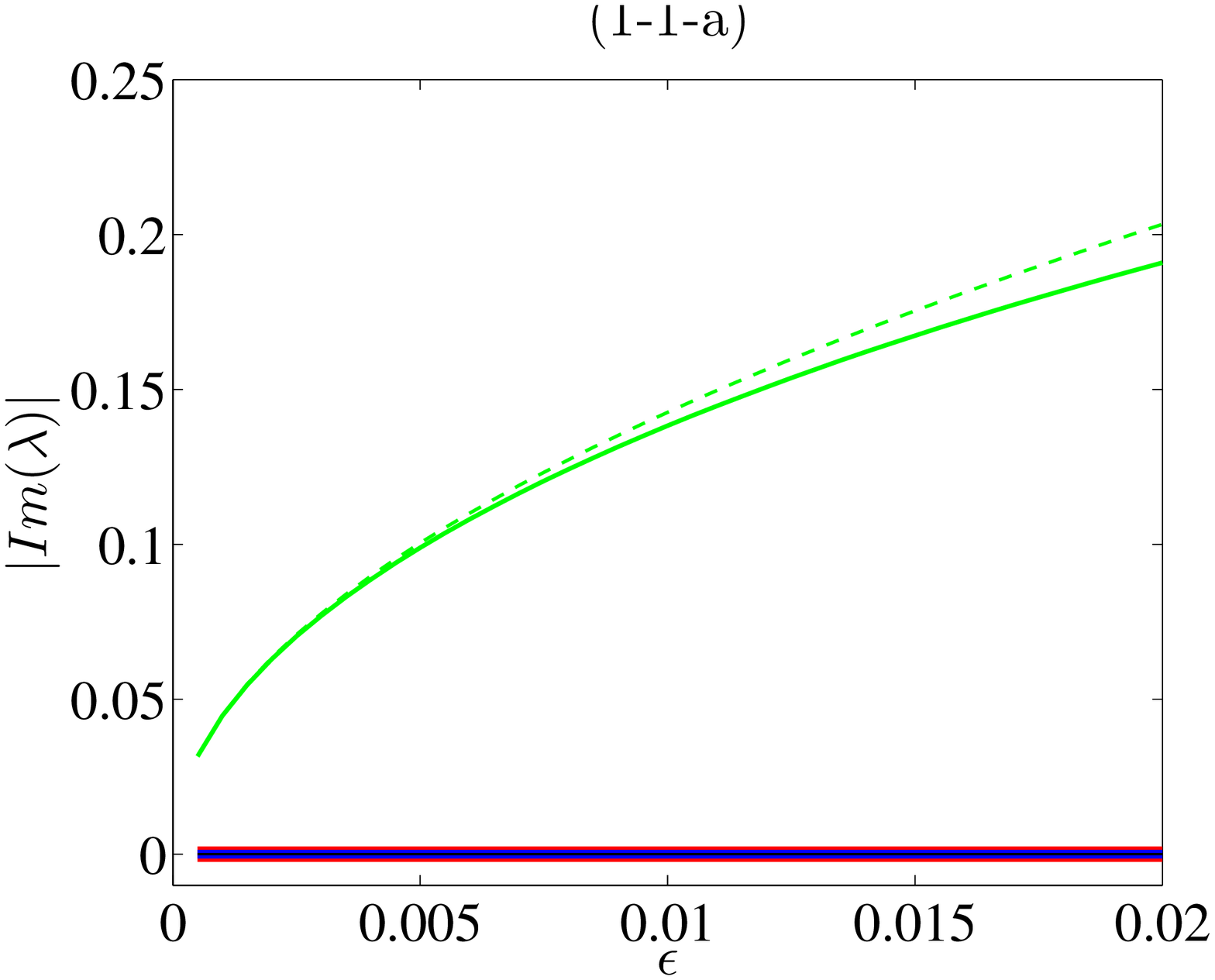}
\end{tabular}
\caption{The real (left) and imaginary (right) parts of eigenvalues $\lambda$
as functions of $\epsilon$ at $\gamma=1+\frac{\sqrt{3}}{2}$ for the branch (1-1-a)
in comparison with the analytically approximated eigenvalues (dash lines). In the left panel, two thinner (green and black) lines are zero and the small real eigenvalue (blue) appears beyond the leading-order asymptotic theory; while in the right panel, all lines but the third thickest (red, blue and black) are identically zero. }
\label{fig1_1}
\end{figure}

\item[(1-3)] For the solution with $2 \theta_2 = -2 \theta_1$ and
$\sin(2 \theta_1) = -\frac{\gamma}{2}$, we obtain $a = b = c = \cos(2 \theta_1)$.
Therefore, the matrix $\mathcal{M}$ has
a simple zero eigenvalue $\mu_1 = 0$ with eigenvector $(1,1,1,1)^T$,
a double nonzero eigenvalue $\mu_2 = \mu_3 = 2 \cos(2\theta_1)$
with eigenvectors $(i,-1,-i,1)^T$ and $(-i,-1,i,1)^T$, and a simple nonzero eigenvalue
$\mu_4 = 4 \cos(2 \theta_1)$ with eigenvector $(-1,1,-1,1)^T$.

Now we can distinguish between the branches (1-3-a) and (1-3-b)
which correspond to $\theta_1 = -\frac{1}{2}\arcsin\left(\frac{\gamma}{2}\right)$ and
$\theta_1 =\frac{\pi}{2}+\frac{1}{2}\arcsin\left(\frac{\gamma}{2}\right)$ respectively.

By Lemma~\ref{lemma1_2}, the spectral stability problem (\ref{ev_matrix2}) for the branch (1-3-a)
has only pairs of real eigenvalues $\lambda$, moreover, the pair of double real eigenvalues in the first-order
reduction may split to a complex quartet beyond the first-order reduction. Independently of
the outcome of this splitting, the branch (1-3-a) is spectrally unstable.

On the other hand, the spectral stability problem (\ref{ev_matrix2}) for the branch (1-3-b)
has only pairs of imaginary eigenvalues $\lambda$ in the first-order reduction in $\epsilon$.
However, one pair of imaginary eigenvalues $\lambda$ is double and may split to a complex quartet beyond
the first-order reduction. If this splitting actually occurs, the branch (1-3-b) is spectrally unstable as well.

Figures \ref{fig1_3} and \ref{fig1_3b} illustrate
the comparisons between numerical eigenvalues and approximated eigenvalues
for the branches (1-3-a) and (1-3-b) respectively. We can see that the pair of double
real eigenvalues $\lambda$ of the spectral stability problem (\ref{ev_matrix2}) for the branch (1-3-a) splits
into two pairs of simple real eigenvalues, hence the complex quartets do not appear as a result
of this splitting. On the other hand, the pair
of double imaginary eigenvalues $\lambda$ of the spectral stability problem (\ref{ev_matrix2}) for the branch (1-3-b)
does split into a quartet of complex eigenvalues, which results in the
(weak) spectral instability of the branch (1-3-b).

\begin{figure}[!htbp]
\centering
\begin{tabular}{c}
\includegraphics[width=7cm]{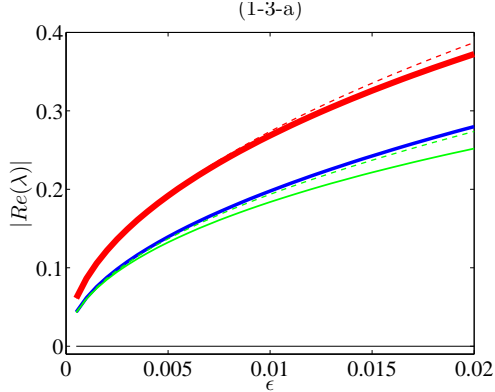}
\end{tabular}
\caption{The real parts of eigenvalues $\lambda$
as functions of $\epsilon$ at $\gamma = -0.7$ for the branch (1-3-a) in comparison with
the analytically approximated eigenvalues (dash lines). The imaginary parts are identically zero.
The double real eigenvalue
splits beyond the leading-order asymptotic theory.}
\label{fig1_3}
\end{figure}

\begin{figure}[!htbp]
\begin{tabular}{ccc}
\includegraphics[width=7cm]{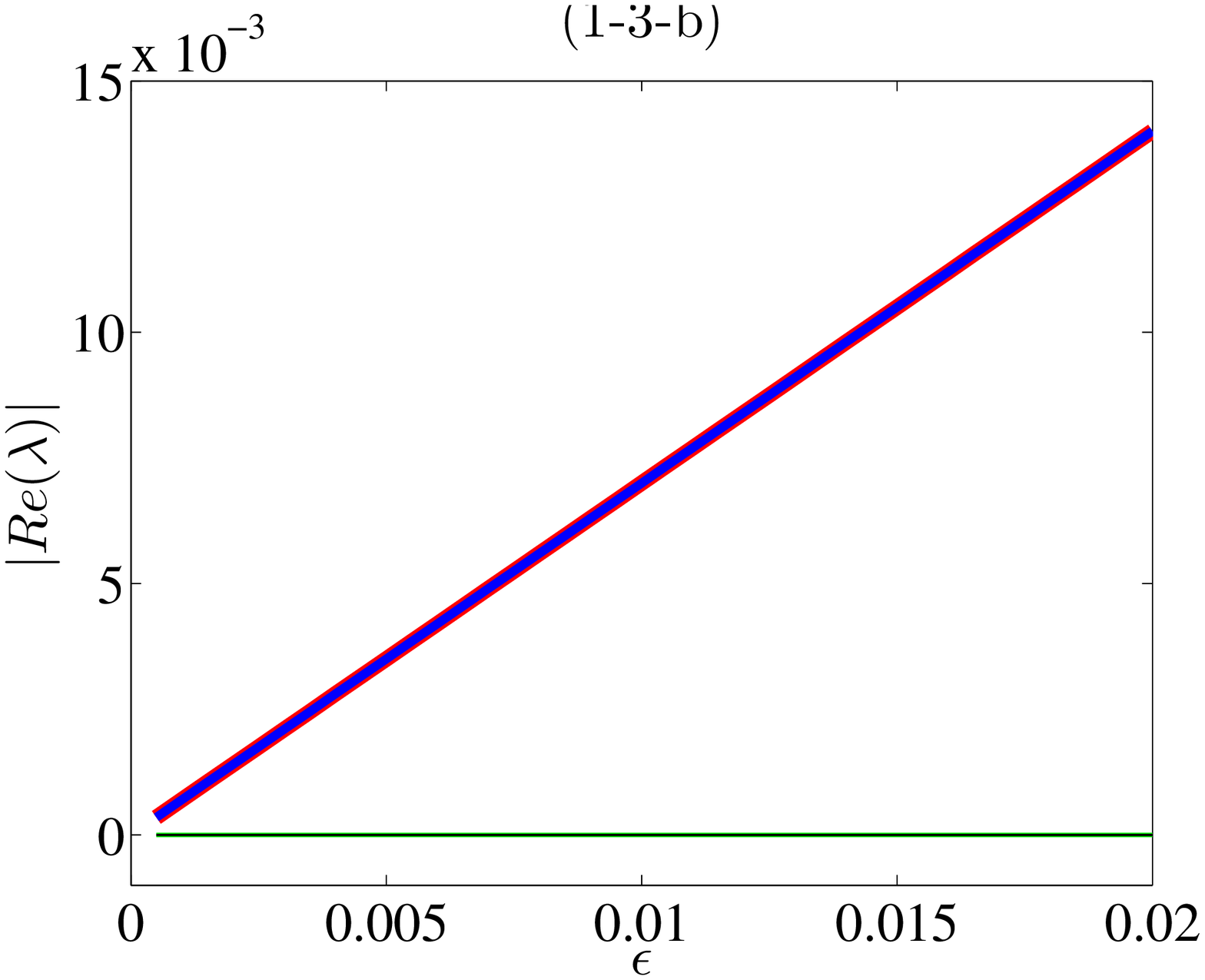}
\includegraphics[width=7cm]{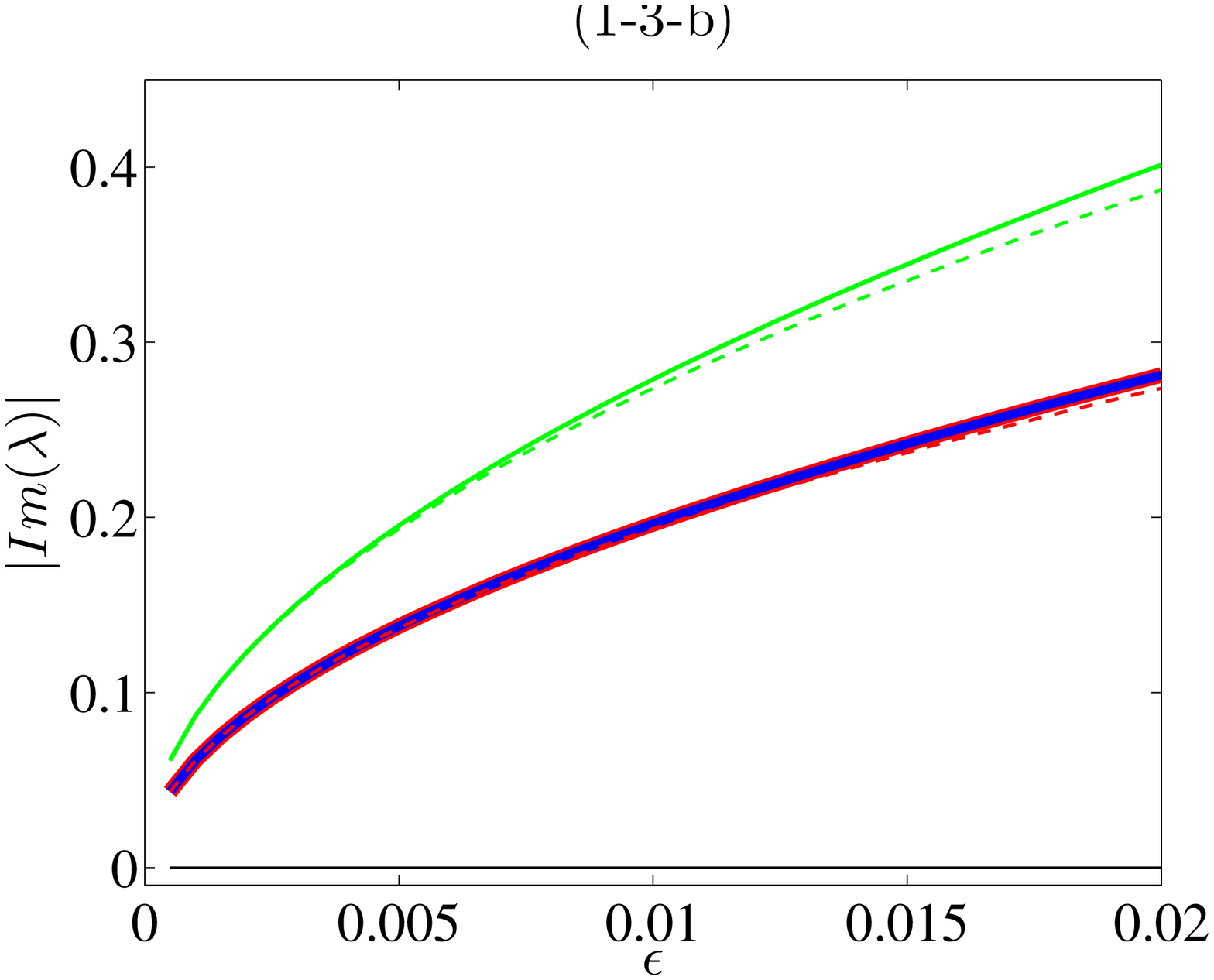}
\end{tabular}
\caption{The real (left) and imaginary (right) parts of eigenvalues $\lambda$
as functions of $\epsilon$ at $\gamma = -0.7$ for the branch (1-3-b)
in comparison with the analytically
approximated eigenvalues (dash lines). In the left panel, two thicker (red and blue) solid lines as well as two thinner (black and green)
solid lines are identical. Similarly, the two thicker (red and blue) solid lines are identical in the right panel. These two sets of coincident
lines correspond to the complex eigenvalue quartet which
emerges in this case, destabilizing the discrete soliton configuration.
Splitting of the double imaginary eigenvalues is beyond the leading-order asymptotic theory.}
\label{fig1_3b}
\end{figure}
\end{itemize}

\subsection{Symmetry about the center (S2)}

Under conditions $\gamma_1 = -\gamma_3$ and $\gamma_2 = -\gamma_4$, we consider the $\pt$-symmetric
configuration in the form $\phi_1=\bar{\phi}_3$ and $\phi_2=\bar{\phi}_4$.
Thus, we have $\theta_1 = -\theta_3$ and $\theta_2 = -\theta_4$, after which
the $4$-by-$4$ matrix $\mathcal{M}$ given by (\ref{eqn0_M}) can be written in the
explicit form:
\begin{equation*}
\mathcal{M}=
\left(
    \begin{array}{cccc}
        a + b & -b & 0 & -a \\
        -b & a + b & -a & 0 \\
        0 & -a & a+b & -b \\
        -a & 0 & -b & a+b
    \end{array}
\right),
\end{equation*}
where $a = \cos(\theta_1+\theta_2)$ and $b = \cos(\theta_1 -\theta_2)$.
Using the transformation matrix
\begin{equation*}
T = \left(
    \begin{array}{cccc}
        1 & 0 & 1 & 0 \\
        0 & 1 & 0 & 1 \\
        1 & 0 & -1 & 0 \\
        0 & 1 & 0 & -1
    \end{array}
\right),
\end{equation*}
which generalizes the matrix $P$ given by (\ref{symmetry-matrix}) for (S2)
we diagonalize the matrix $M$ into two blocks after a similarity transformation:
\begin{equation*}
T^{-1} M T = \left(
    \begin{array}{cccc}
        a+b & -a-b & 0 & 0 \\
        -a-b & a+b & 0 & 0 \\
        0 & 0 & a + b & a-b \\
        0 & 0 & a-b & a + b
    \end{array}
\right).
\end{equation*}
Note that the first block is given by a singular matrix, whereas the second block coincides
with the Jacobian matrix $\mathcal{N}(\boldsymbol{\theta})$ given by (\ref{eqn3_N}).

The following list summarizes stability of the branches (2-1) and (2-2) among the solutions of the system (\ref{eqn3_4}),
which corresponds to the particular case $\gamma_1 = -\gamma_2 = \gamma$.

\begin{itemize}
\item[(2-1)] For the solution with $\theta_2 = -\theta_1$ and $\sin(2\theta_1)=-\gamma$,
we obtain $a = 1$ and $b = \cos(2\theta_1)$. Therefore,
the matrix $\mathcal{M}$ has zero eigenvalue $\mu_1 = 0$ with eigenvector $(1,1,1,1)^T$
and three simple nonzero eigenvalues $\mu_2 = 2$ with eigenvector $(-1,-1,1,1)^T$,
$\mu_3 = 2 \cos((2\theta_1)$ with eigenvector $(1,-1,-1,1)^T$,
and $\mu_4 = 2 + 2 \cos(2\theta_1)$ with eigenvector $(-1,1,-1,1)^T$.

For $\epsilon > 0$, the spectral problem (\ref{ev_matrix2}) has at least one pair of
real eigenvalues $\lambda$ near $\pm\sqrt{4\epsilon}$ so that the stationary solutions are
spectrally unstable for both branches (2-1-a) and (2-1-b).
The numerical results shown in Figures \ref{fig2_1} and \ref{fig2_1b}
illustrate the validity of the first-order approximations for the eigenvalues
of the spectral stability problem (\ref{ev_matrix2}).

\begin{figure}[!htbp]
\centering
\begin{tabular}{ccc}
\includegraphics[width=7cm]{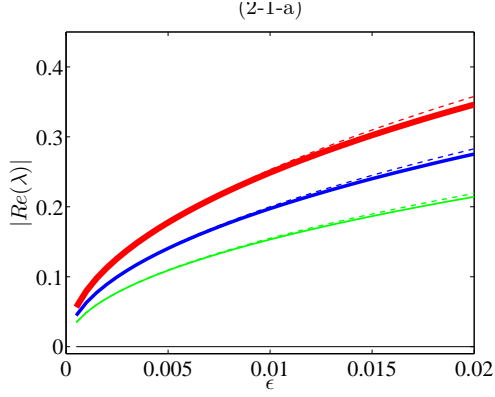}
\end{tabular}
\caption{The real parts of eigenvalues $\lambda$ as functions of $\epsilon$ at $\gamma = 0.8$ for the branch (2-1-a) in comparison with the
analytically approximated eigenvalues (dash lines). The imaginary parts are identically zero, i.e., all
three nonzero eigenvalue pairs are real.}
\label{fig2_1}
\end{figure}

\begin{figure}[!htbp]
\begin{tabular}{ccc}
\includegraphics[width=7cm]{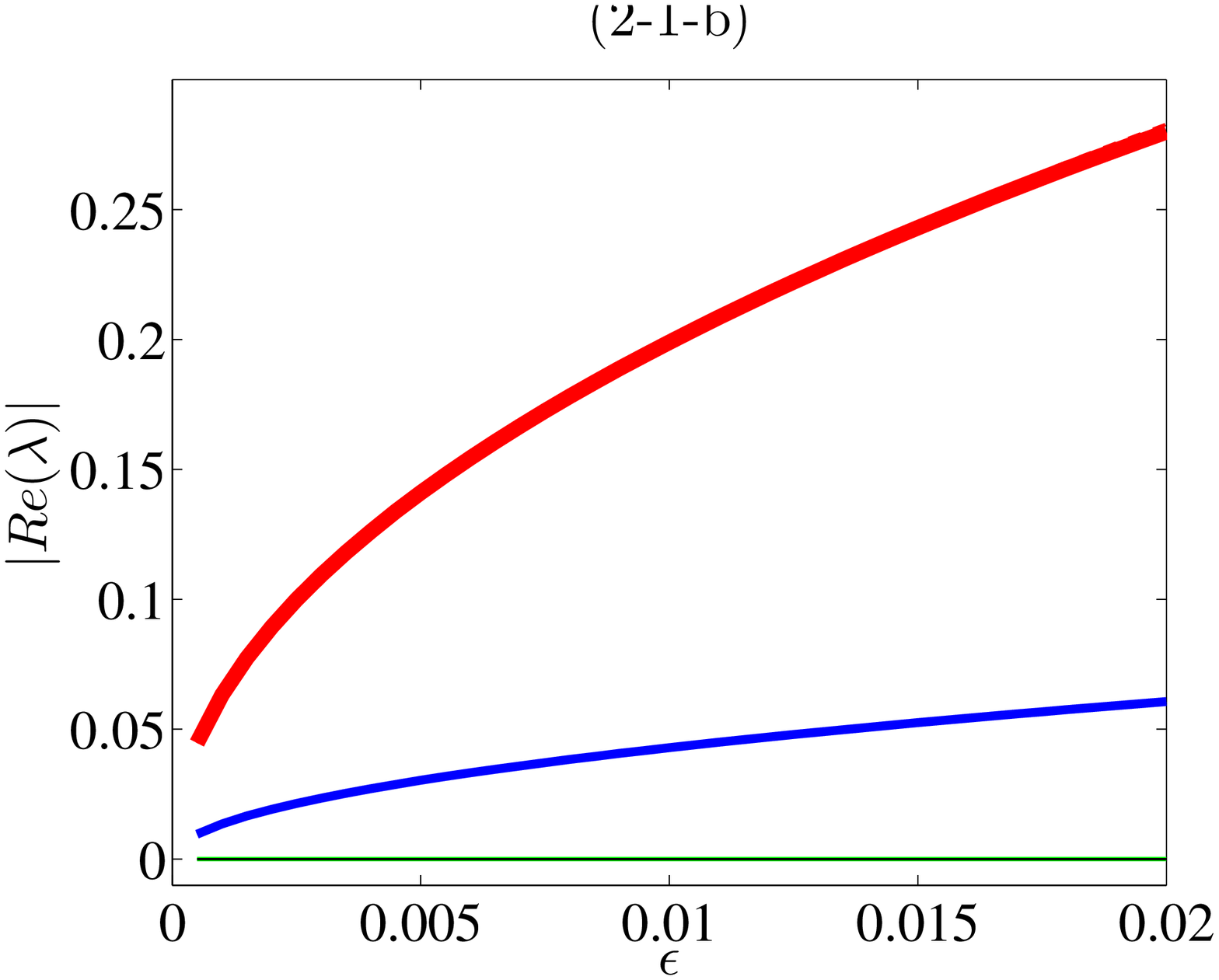}
\includegraphics[width=7cm]{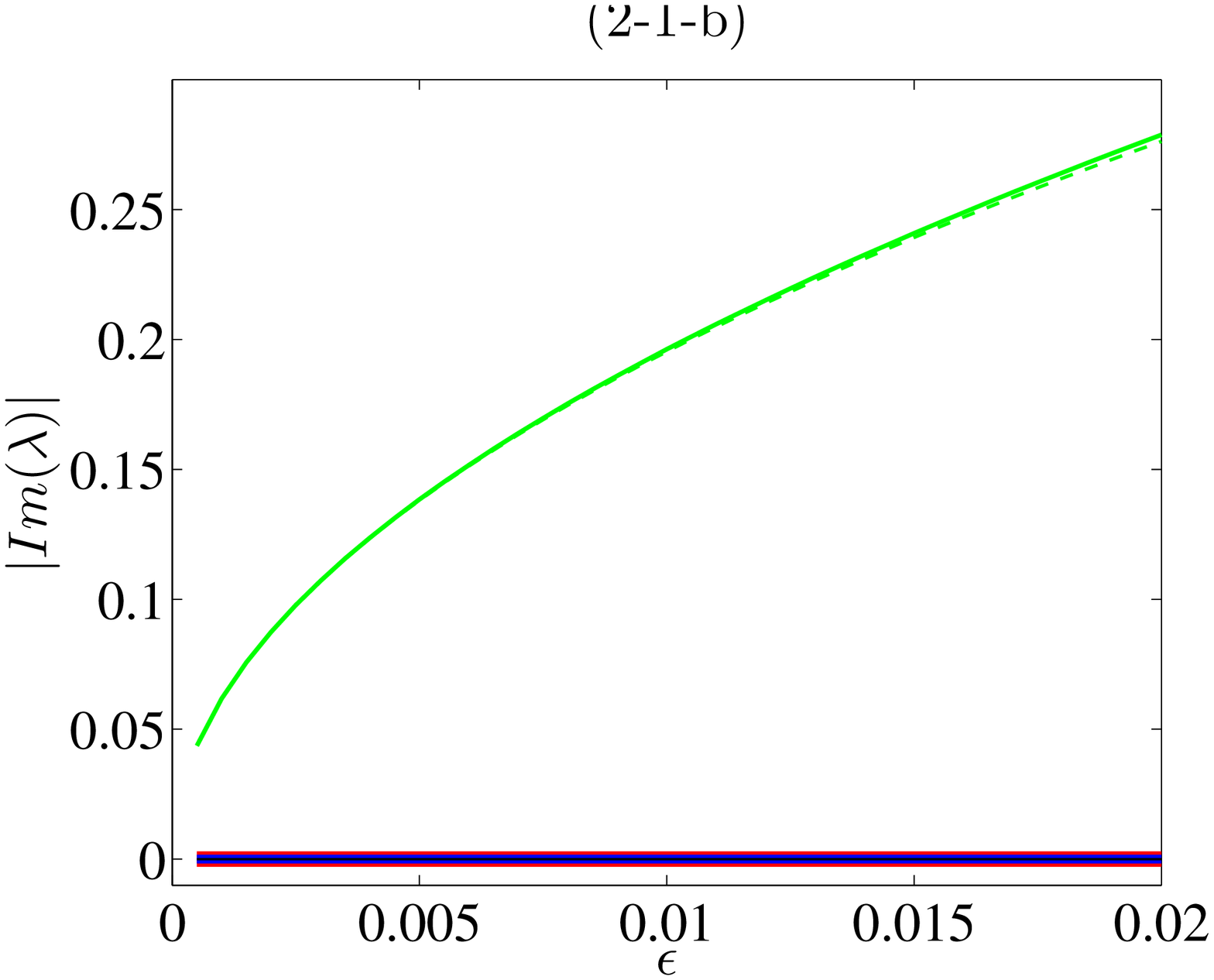}
\end{tabular}
\caption{The real and imaginary parts of eigenvalues $\lambda$
as functions of $\epsilon$ at $\gamma = -0.3$ for the branch (2-1-b) in comparison with
the analytically approximated eigenvalues (dash lines). In the left panel two thinner (green and black) solid lines are zero while in the right panel all solid lines but the third thickest (red, blue and black) are zero. Here, two
of the relevant eigenvalue pairs are found to be real, while the other
is imaginary, in good agreement with the theoretical predictions. In both panels, the solid lines and dash lines look almost identical.}
\label{fig2_1b}
\end{figure}

\item[(2-2)] For the solution with $\theta_2 = -\theta_1 \pm \pi$ and $\sin(2\theta_1) = \gamma$,
we obtain $a = -1$ and $b = -\cos(2 \theta_1)$. Therefore,
the matrix $\mathcal{M}$ has zero eigenvalue $\mu_1 = 0$ with eigenvector $(1,1,1,1)^T$
and three simple nonzero eigenvalues $\mu_2 = -2$, $\mu_3 = -2 \cos(2\theta_1)$,
and $\mu_4 = 2 + 2 \cos(2\theta_1)$. These eigenvalues are opposite to those in the case (2-1),
because the family (2-2) is related to the family (2-1) by Remark \ref{remark1_2}. Consequently,
the stability analysis of the family (2-2) for $\epsilon > 0$ corresponds to
the stability analysis of the family (2-1) for $\epsilon < 0$.

For the branch (2-2-a), the spectral stability problem (\ref{ev_matrix2}) with $\epsilon > 0$
has three pairs of simple purely imaginary eigenvalues $\lambda$ near $\pm 2i\sqrt{\epsilon}$, $\pm 2 i\sqrt{\epsilon \cos(2\theta_1)}$,
and $\pm 2 i\sqrt{\epsilon (1+\cos(2\theta_1))}$. Therefore, the stationary solution is spectrally stable
at least for small values of $\epsilon > 0$. For the branch (2-2-b), the spectral stability problem (\ref{ev_matrix2}) with $\epsilon > 0$
include a pair of real eigenvalues near $\pm 2 \sqrt{|\epsilon \cos(2\theta_1)|}$ , which implies instability of the stationary solutions.
The numerical results shown in Figures \ref{fig2_2} and \ref{fig2_2b}
illustrate the validity of these predictions.

\begin{figure}[!htbp]
\centering
\begin{tabular}{ccc}
\includegraphics[width=7cm]{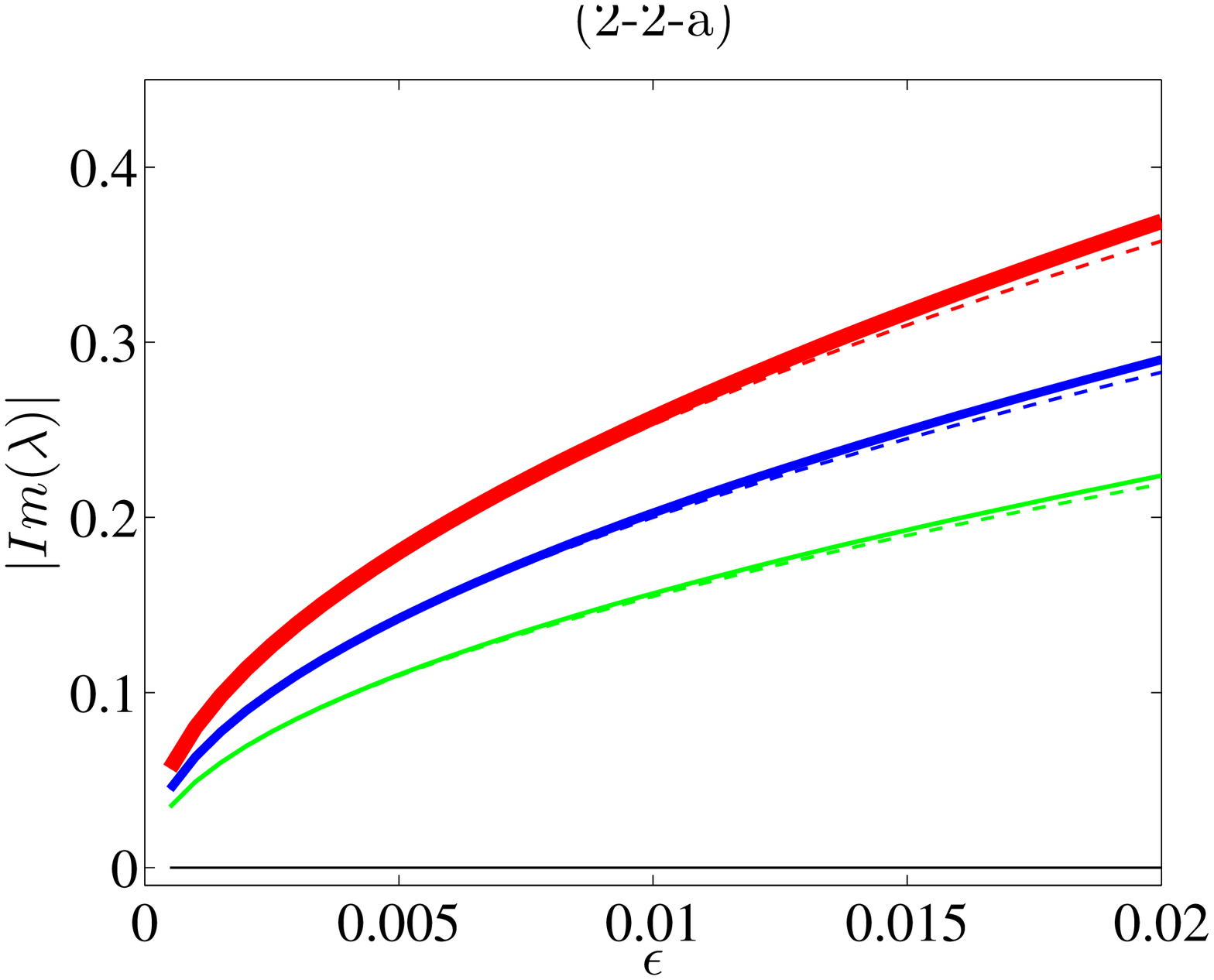}
\end{tabular}
\caption{
The imaginary parts of eigenvalues $\lambda$ as functions of $\epsilon$ at $\gamma = 0.8$ for the
stable branch (2-2-a) in comparison with the analytically
approximated eigenvalues (dash lines). The real parts of eigenvalues are identically zero. }
\label{fig2_2}
\end{figure}

\begin{figure}[!htbp]
\begin{tabular}{ccc}
\includegraphics[width=7cm]{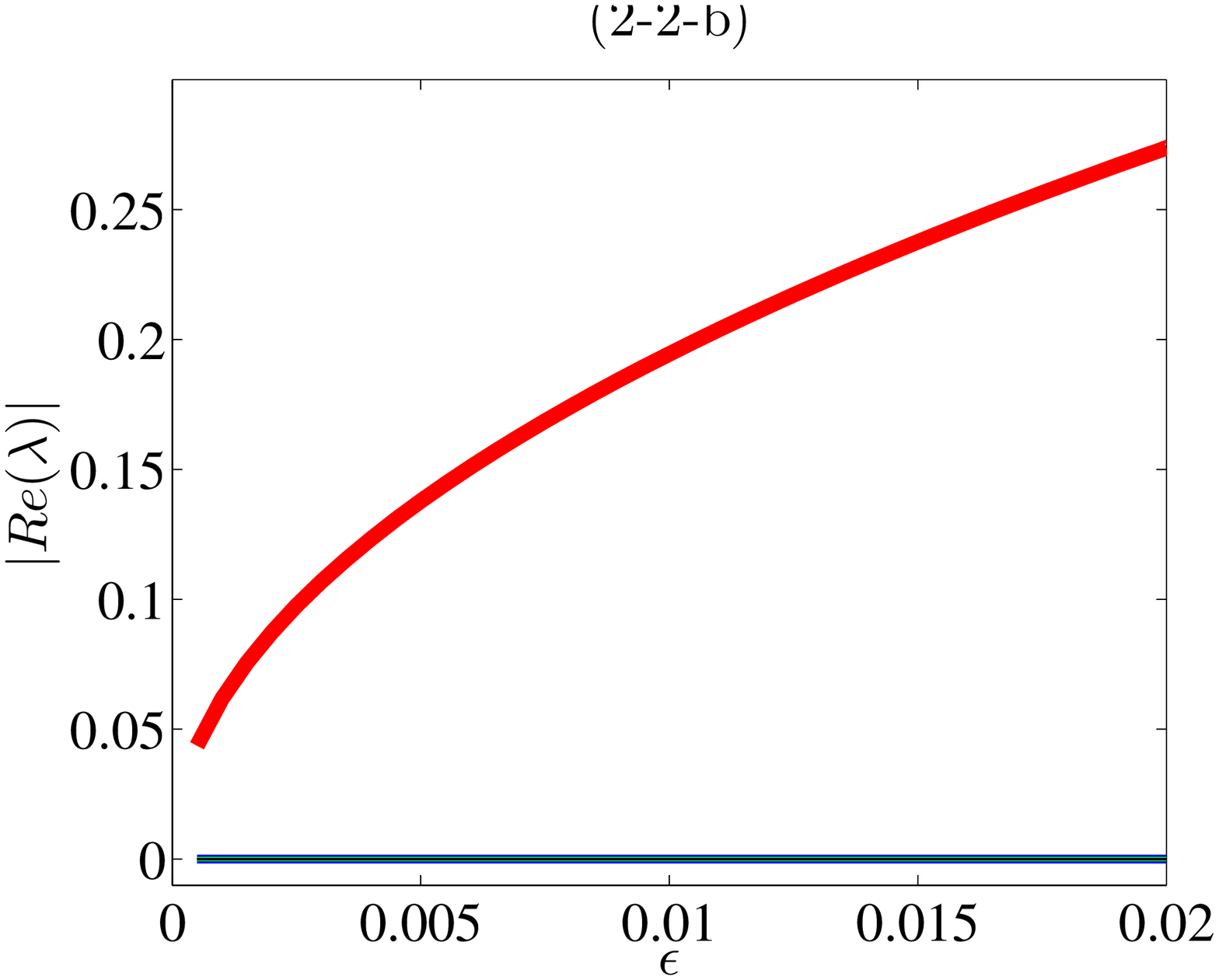}
\includegraphics[width=7cm]{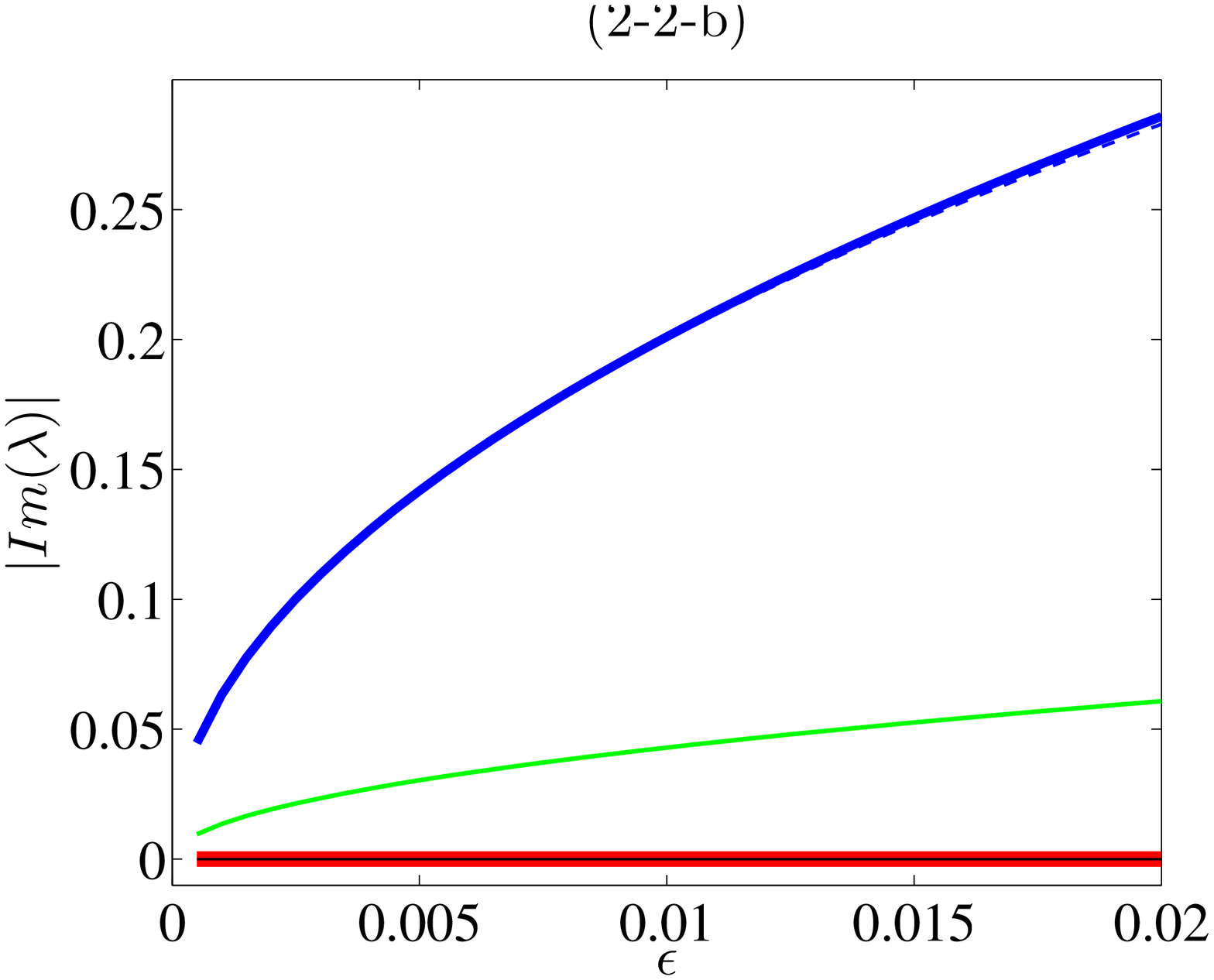}
\end{tabular}
\caption{The real parts of eigenvalues $\lambda$ as functions of $\epsilon$ at $\gamma = 0.3$ for the branch (2-2-b)
in comparison with the approximated eigenvalues (dash lines). In the left panel three thinner (blue, green and black) solid lines are zero while in the right panel the thickest and the thinnest (red and black) solid lines are zero. Here,
one of the three pairs is found to be real, while the other two
are imaginary. In both panels, the solid and dash lines look almost identical.}
\label{fig2_2b}
\end{figure}

\end{itemize}

\section{Stability of $\pt$-symmetric configurations in truncated lattice}
\label{section-stability-extended}

Here we extend the spectral stability analysis of the $\pt$-symmetric configurations
to the setting of the truncated square lattice. Persistence of these configurations
in small parameter $\epsilon$ is obtained in Section \ref{section-full}.

Let $n$ be an even number and consider the $n$-by-$n$
square lattice, denoted by $L_n$. The domain is truncated symmetrically with zero boundary
conditions. The $\pt$-symmetric solution $\{\phi_{j,k}\}_{(j,k) \in L_n}$
to the system (\ref{eqn0_2}) is supposed to satisfy
the limiting configuration (\ref{limiting-config-lattice}), where
the central cell $S$ is now placed at
$$
S = \left\{ \left(\frac{n}{2}, \frac{n}{2} \right),
\left( \frac{n}{2}, \frac{n}{2}+1 \right),
\left( \frac{n}{2}+1, \frac{n}{2} \right),
\left( \frac{n}{2}+1, \frac{n}{2}+1 \right) \right\},
$$
whereas the zero sites are located at $S^* := L_n \backslash S$.

The spectral stability problem is given by (\ref{ev_matrix2}), where
$\boldsymbol{\xi}$ consists of blocks of $(v_{j,k},w_{j,k})^T$,
$\mathcal{G}$ consists of blocks of $\gamma_{j,k} \sigma_3$,
and $\mathcal{H}$ consists of the blocks of
\begin{align}
\label{eqn0_H_2}
\mathcal{H}_{j,k}=
\left(
    \begin{array}{cc}
        1-2|\phi_{j,k}|^2 & -\phi_{j,k}^2\\
        -(\overbar{\phi_{j,k}})^2 & 1-2|\phi_{j,k}|^2
    \end{array}
\right)
-\epsilon(s_{0,+1}+s_{0,-1}+s_{-1,0}+s_{+1,0})
\left(
    \begin{array}{cc}
        1 & 0\\
        0 & 1
    \end{array}
\right),
\end{align}
where $s_{l,m}$ stands for the shift operator such that $s_{l,m} \phi_{j,k} = \phi_{j+l,k+m}$.

The limiting operator $\mathcal{H}^{(0)}$ at $\epsilon = 0$ has two semi-simple eigenvalues
$\mu = 0$ and $\mu = -2$ of multiplicity four and a semi-simple eigenvalue $\mu = 1$ of multiplicity
$2n^2-8$. On the other hand, the spectral stability problem (\ref{ev_matrix2}) for $\epsilon = 0$
has the zero eigenvalue $\lambda = 0$ of geometric multiplicity four and algebraic multiplicity eight
and a pair of semi-simple eigenvalues $\lambda = \pm i$ of multiplicity $n^2 - 4$.
When $\epsilon$ is nonzero but small, the splitting of the zero eigenvalue $\lambda = 0$
is the same as that on the elementary cell $S$ if the splitting occurs
in the first-order perturbation theory. However, unless $\gamma_{j,k}=0$ for all $(j,k) \in S^*$,
the splitting of the semi-simple eigenvalues $\lambda = \pm i$ is non-trivial and
can possibly lead to instability, as is shown in the analysis of \cite{pel3}.

In order to study the splitting of the semi-simple eigenvalues $\lambda = \pm i$ for small but nonzero $\epsilon$,
we expand $\boldsymbol{\xi}=\boldsymbol{\xi}^{(0)}+\epsilon \boldsymbol{\xi}^{(1)}+\mathcal{O}(\epsilon^2)$  and
$\lambda=\lambda^{(0)}+\epsilon \lambda^{(1)}+\mathcal{O}(\epsilon^2)$ where $\lambda^{(0)} = \pm i$,
and obtain the perturbation equations
\begin{eqnarray}
\label{eqn_ev0}
 \mathcal{H}^{(0)} \boldsymbol{\xi}^{(0)} = i\lambda^{(0)}\sigma\boldsymbol{\xi}^{(0)},
\end{eqnarray}
and
\begin{eqnarray}
\label{eqn_ev1}
 (\mathcal{H}^{(1)}+i\mathcal{G}) \boldsymbol{\xi}^{(0)} +  \mathcal{H}^{(0)} \boldsymbol{\xi}^{(1)} = i\lambda^{(0)}\sigma\boldsymbol{\xi}^{(1)}+i\lambda^{(1)}\sigma\boldsymbol{\xi}^{(0)}.
\end{eqnarray}
Since $i\lambda^{(0)}\in\mathbb{R}$, the operator $(\mathcal{H}^{(0)}-i\lambda^{(0)}\sigma)$ is self-adjoint,
and we assume that $\ker(\mathcal{H}^{(0)}-i\lambda^{(0)}\sigma)$ is spanned by $\{\boldsymbol{\psi}_{j,k}\}_{(j,k)\in S^*}$.
Then, $\boldsymbol{\xi}^{(0)}=\sum\limits_{(j,k)\in S^*} c_{j,k} \boldsymbol{\psi}_{j,k}$
satisfies (\ref{eqn_ev0}). Projection of (\ref{eqn_ev1}) to $\ker(\mathcal{H}^{(0)}-i\lambda^{(0)}\sigma)$ yields
the matrix eigenvalue problem
\begin{eqnarray}
\label{eqn_ev1_2}
 i\lambda^{(1)}D\boldsymbol{c} = K \boldsymbol{c},
\end{eqnarray}
where
\begin{eqnarray*}
D_{\mathcal{P}(j_1,k_1),\mathcal{P}(j_2,k_2)} & = & \langle \boldsymbol{\psi}_{j_1,k_1}, \sigma\boldsymbol{\psi}_{j_2,k_2} \rangle, \\
K_{\mathcal{P}(j_1,k_1),\mathcal{P}(j_2,k_2)} & = & \langle \boldsymbol{\psi}_{j_1,k_1},
(\mathcal{H}^{(1)}+i\mathcal{G})\boldsymbol{\psi}_{j_2,k_2} \rangle,
\end{eqnarray*}
for $(j_1,k_1), (j_2,k_2)\in S^*$, and the bijective mapping $\mathcal{P}: S^*\to \{1,2,...,n^2-4\}$ is defined by
\begin{eqnarray}
\label{eqn_K}
\mathcal{P}(j,k)=
\begin{cases}
(j-1)n+k, &(j-1)n+k< (\frac{n}{2}-1)n+\frac{n}{2}\cr
(j-1)n+k-2, &(\frac{n}{2}-1)n+\frac{n}{2}+1< (j-1)n+k< \frac{n^2}{2}+\frac{n}{2}\cr
(j-1)n+k-4, &(j-1)n+k> \frac{n^2}{2}+\frac{n}{2}+1
\end{cases}
\end{eqnarray}
assuming that the lattice $L_n$ is traversed in the order
$$
(1,1),(1,2),...,(1,n),(2,1),(2,2),...,(2,n),...,(n,1),(n,2),...,(n,n).
$$

For $i\lambda^{(0)}=1$, the eigenvector $\boldsymbol{\psi}_{j,k}$ has the only nonzero block
$(1,0)^T$ at $\mathcal{P}(j,k)$-th position which corresponds to position $(j,k)$ on the lattice.
Therefore, we have $D=I_{n^2-4}$ and
\begin{eqnarray}
\label{eqn_M_near1}
K_{\mathcal{P}(j_1,k_1),\mathcal{P}(j_2,k_2)}=
\begin{cases}
i\gamma_{j_1,k_1}, &(j_1,k_1)=(j_2,k_2)\cr
-1, &|j_1-j_2|+|k_1-k_2|=1\cr
0, &{\rm otherwise}
\end{cases}
\end{eqnarray}
For $i\lambda^{(0)}=-1$, the eigenvector $\boldsymbol{\psi}_{j,k}$ has the only nonzero block
$(0,1)^T$ at $\mathcal{P}(j,k)$-th position, then $D = -I_{n^2-4}$ and $K$ is the complex conjugate of
$K$ given by (\ref{eqn_M_near1}). As a result, the eigenvalues of the reduced eigenvalue problem (\ref{eqn_ev1_2})
for $i\lambda^{(0)}=-1$ are complex conjugate to those for $i\lambda^{(0)}=1$.

As an example, matrix $K$ (for $i\lambda^{(0)}=1$) is obtained for $n=4$ in the following form:
\begin{equation}
\label{S-example}
K = \left(
    \begin{array}{cccccccccccc}
        i\gamma_{1,1} & -1 & 0 & 0 & -1 & 0 & 0 & 0 & 0 & 0 & 0 & 0 \\
        -1 & i\gamma_{1,2} & -1 & 0 & 0 & 0 & 0 & 0 & 0 & 0 & 0 & 0 \\
        0 & -1 & i\gamma_{1,3} & -1 & 0 & 0 & 0 & 0 & 0 & 0 & 0 & 0 \\
        0 & 0 & -1 & i\gamma_{1,4} & 0 & -1 & 0 & 0 & 0 & 0 & 0 & 0 \\
        -1 & 0 & 0 & 0 & i\gamma_{2,1} & 0 & -1 & 0 & 0 & 0 & 0 & 0 \\
        0 & 0 & 0 & -1 & 0 & i\gamma_{2,4} & 0 & -1 & 0 & 0 & 0 & 0 \\
        0 & 0 & 0 & 0 & -1 & 0 & i\gamma_{3,1} & 0 & -1 & 0 & 0 & 0 \\
        0 & 0 & 0 & 0 & 0 & -1 & 0 & i\gamma_{3,4} & 0 & 0 & 0 & -1 \\
        0 & 0 & 0 & 0 & 0 & 0 & -1 & 0 & i\gamma_{4,1} & -1 & 0 & 0 \\
        0 & 0 & 0 & 0 & 0 & 0 & 0 & 0 & -1 & i\gamma_{4,2} & -1 & 0 \\
        0 & 0 & 0 & 0 & 0 & 0 & 0 & 0 & 0 & -1 & i\gamma_{4,3} & -1 \\
        0 & 0 & 0 & 0 & 0 & 0 & 0 & -1 & 0 & 0 & -1 & i\gamma_{4,4} \\
    \end{array}
\right).
\end{equation}

The $\pt$-symmetric configurations for the symmetries (S1) and (S2)
correspond to the cases $\gamma_{j,k}=(-1)^{j+k}\gamma$
and $\gamma_{j,k} = (-1)^j \gamma$ for all $(j,k)\in L_n$, where $\gamma \in \mathbb{R}$.
In the first case, we shall prove that the eigenvalues $\lambda^{(1)}$
of the reduced eigenvalue problem (\ref{eqn_ev1_2}) include
a pair of real eigenvalues for every $\gamma \neq 0$.
In the second case, we shall show numerically that the eigenvalues
$\lambda^{(1)}$ of the reduced eigenvalue problem (\ref{eqn_ev1_2})
remain purely imaginary for sufficiently small $\gamma \neq 0$.

\begin{proposition}
\label{proposition5_1}
Let $\gamma_{j,k}=(-1)^{j+k}\gamma$ for all $(j,k) \in L_n$ and $K$ be given by (\ref{eqn_M_near1}).
If we define $A=\re(K)$ and $B=\im(K|_{\gamma=1})$ then $A$ and $B$ are anti-commute, i.e. $AB = -BA$.
\end{proposition}

\begin{proof}
By inspecting the definition of matrices in (\ref{eqn_M_near1}) with $\gamma_{j,k}=(-1)^{j+k}\gamma$
for all $(j,k) \in L_n$, we observe that $A$ is symmetric, $B$ is diagonal, and $B^2 = I_{n^2-4}$. In calculating $AB$,
the $i$-th column of $A$ is multiplied by $i$-th entry in the diagonal of $B$, while the $i$-th row of $A$
is multiplied by $i$-th entry in the diagonal of $B$ in calculating $BA$. In the lattice $S^* \subset L_n$,
each site at $(j_1,k_1)$ with $\gamma_{j_1,k_1} = \gamma$ is surrounded by sites $(j_2,k_2)$ with $\gamma_{j_2,k_2} = -\gamma$
and vice verse, so that each nonzero entry in $A$ must sit in a position $(\mathcal{P}(j_1,k_1), \mathcal{P}(j_2,k_2))$, where $\gamma_{j_1,k_1}=-\gamma_{j_2,k_2}$. Based on the facts above, each nonzero entry in $A$ changes its sign either
in $AB$ or in $BA$, hence $AB = -BA$.
\end{proof}

\begin{proposition}
\label{lemma5_2}
Under the assumptions of Proposition~\ref{proposition5_1}, $A$ has a zero eigenvalue of multiplicity at least $n-2$.
\end{proposition}

\begin{proof}
At first, we assume the lattice is $L_n$ and consider the spectral problem $\tilde{A} u = \lambda u$, where
\begin{eqnarray}
\label{Ln_A_spectral}
\tilde{A}_{(j_1-1)n+k_1,(j_2-1)n+k_2}=
\begin{cases}
-1, &|j_1-j_2|+|k_1-k_2|=1\cr
0, &{\rm otherwise}
\end{cases}
\end{eqnarray}
for $(j_{1,2},k_{1,2})\in L_n$. In fact, this problem represents the difference equations
\begin{equation}
\label{discrete-Laplacian}
-u_{j,k+1} - u_{j,k-1} - u_{j+1,k} - u_{j-1,k} = \lambda u_{j,k}, \quad (j,k) \in L_n,
\end{equation}
which are closed with the Dirichlet end-point conditions.
This spectral problem can be solved with the double discrete Fourier transform, from which
we obtain the eigenvalues
$$
\lambda_{l,m} = -2 \cos(\frac{l\pi}{n+1}) - 2 \cos(\frac{m\pi}{n+1})
$$
and the eigenvectors
$$
u_{l,m}(j,k) = \sin(\frac{j l\pi}{n+1}) \sin(\frac{k m\pi}{n+1})
$$
for $1 \leq l,m \leq n$. Thus, the spectral problem (\ref{discrete-Laplacian})
has $n^2$ eigenvalues, among which $n$ eigenvalues are zero and they correspond to $l + m = n+1$.

Now, the spectral problem $A u = \lambda u$
is actually posed in $S^* := L_n \backslash S$, which is different from the spectral problem
(\ref{discrete-Laplacian}) by adding four constraints $u(j,k) = 0$ for $(j,k) \in S$.
A linear span of $n$ linearly independent eigenvectors for the zero eigenvalue of multiplicity $n$
satisfies the four constraints at the subspace, whose dimension is at least $n-4$. In fact,
it can be directly checked that
$$
u_{l,n+1-l}(j,k) = (-1)^{k+1} \sin(\frac{j l\pi}{n+1}) \sin(\frac{k l\pi}{n+1})=
-u_{l,n+1-l}(n+1-j,n+1-k)
$$
for any $1\leq j,k,l\leq n$. In particular, we notice that the $n\times 4$ matrix consisting of
$u_{l,n+1-l}(j,k)$, where $1\leq l\leq n$ and $(j,k)\in S$, is of rank $2$. Therefore,
linear combinations of $\{u_{l,n+1-l}\}_{l=1}^n$ that satisfy the constraints $u(j,k) = 0$ for $(j,k) \in S$
form a subspace of dimension $n-2$ and the proposition is proved.
\end{proof}

\begin{lemma}
\label{lemma4_3}
Let $\gamma_{j,k}=(-1)^{j+k}\gamma$ for all $(j,k) \in L_n$, $D = I_{n^2-4}$, and $K$ be given by (\ref{eqn_M_near1}).
Then, there is at least one pair of real eigenvalues $\lambda^{(1)}$ of the reduced eigenvalue problem
(\ref{eqn_ev1_2}) for every $\gamma \neq 0$.
\end{lemma}

\begin{proof}
By Proposition \ref{proposition5_1}, we can write $K=A+i\gamma B$ where $B^2=I_{n^2-4}$ and $AB=-BA$.
Squaring the reduced eigenvalue problem (\ref{eqn_ev1_2}), we obtain
$$
K^2\boldsymbol{c} = \left( A^2 - \gamma^2 I_{n^2-4}\right) \boldsymbol{c} = -(\lambda^{(1)})^2\boldsymbol{c}.
$$
Since $A$ is symmetric, eigenvalues $\lambda^{(1)}$ are all purely imaginary for $\gamma = 0$. Nonzero
eigenvalues $\lambda^{(1)}$ remain nonzero and purely imaginary for sufficiently small $\gamma \neq 0$,
since $K^2$ is symmetric and real. On the other hand, by Proposition \ref{lemma5_2},
$A$ always has a zero eigenvalue ($n\geq 4$), hence
$K^2$ has a negative eigenvalue $-\gamma^2$, which corresponds
to a pair of real eigenvalues $\lambda^{(1)} = \pm \gamma$.
\end{proof}

Coming back to the branches (1-1), (1-2), and (1-3) of the $\pt$-symmetric configurations (S1) studied in Section 2.1,
they correspond to the case $\gamma_{j,k}=(-1)^{j+k}\gamma$ for all $(j,k)\in L_n$. By Lemma \ref{lemma4_3},
the reduced eigenvalue problem (\ref{eqn_ev1_2}) has at least one pair of real eigenvalues $\lambda^{(1)}$.
Therefore, all $\pt$-symmetric configurations are unstable on the truncated square lattice $L_n$ because of the
sites in the set $S^*$. In addition, all the configurations are also unstable because of the sites in
the central cell $S$, as explained in Section 4.1.

For the branches (2-1) and (2-2) of the $\pt$-symmetric configurations (S2) studied in Section 2.2,
they correspond to $\gamma_{j,k}=(-1)^{j}\gamma$ for all $(j,k) \in L_n$. We have checked numerically
for all values of $n$ up to $n = 20$ that the reduced eigenvalue problem (\ref{eqn_ev1_2}) has
all eigenvalues $\lambda^{(1)}$ on the imaginary axis,
at least for small values of $\gamma$.

Let $\gamma_c(n)$ be the largest value of $|\gamma|$, for which the eigenvalues of
$S$ are purely real. Figure \ref{fig4_split} shows
how $\gamma_n$ changes with respect to even $n$ from numerical computations. It is seen from the figure that
$\gamma_c(n)$ decreases towards $0$ as $n$ grows.
The latter property agrees with the analytical predictions in \cite{lett2} for unbounded $\pt$-symmetric lattices
with spatially extended gains and losses.

\begin{figure}[!htbp]
\centering
\begin{tabular}{ccc}
\includegraphics[width=7cm]{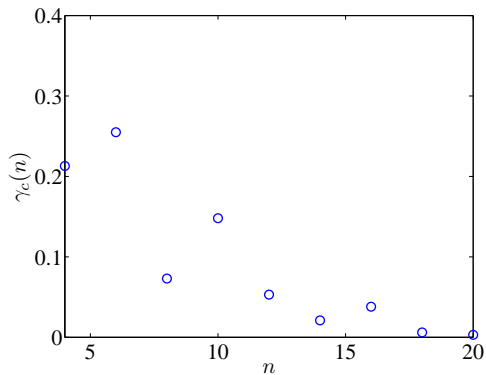}
\end{tabular}
\caption{The threshold $\gamma_c$ versus even $n$ for existence of stable eigenvalues in the reduced eigenvalue
 problem (\ref{eqn_ev1_2}) with $\gamma_{j,k}=(-1)^{j}\gamma$ for all $(j,k) \in L_n$. Accounting also for numerical errors,
here we call a numerically computed eigenvalue stable
if the absolute value of its real part is less than $10^{-7}$.}
\label{fig4_split}
\end{figure}

Besides the eigenvalue computations on the sites of $S^*$, one needs to add eigenvalue computations
on the sites of $S$ performed in Section 4.2. It follows from the results reported there that
only the branch (2-2-a) is thus potentially stable, see Figure \ref{fig2_2}.
The branch of $\pt$-symmetric configurations
remain stable on the extended square lattice $L_n$, provided $|\gamma| < \gamma_c(n)$.
In Figure \ref{fig4_1}, we give an example of the stable $\pt$-symmetric solution on
the $20$-by-$20$ square lattice that continues from the branch (2-2-a) on the elementary cell
$S$. When $\epsilon=0$, the phases of the solution on $S$ in (2-2-a) are
$\{ \theta_1, \pi-\theta_1, -\theta_1, \theta_1-\pi \}$,
where $\theta_1 =\frac{1}{2}\arcsin(\gamma)$, representing a discrete
soliton in the form of an
anti-symmetric (sometimes, called twisted) localized
mode~\cite{dnlsbook} if $\gamma = 0$. When $\epsilon = 0.02$ is small,
we can observe on Figure \ref{fig4_1} that the solution
is still close to the limiting solution at $\epsilon=0$ in terms of
amplitude and phase. Besides, the spectrum in the bottom right panel
of Figure \ref{fig4_1} verifies our expectation about its spectral stability,
bearing eigenvalues solely on the imaginary axis.

\begin{figure}[htbp]
\begin{tabular}{cc}
\includegraphics[width=7cm]{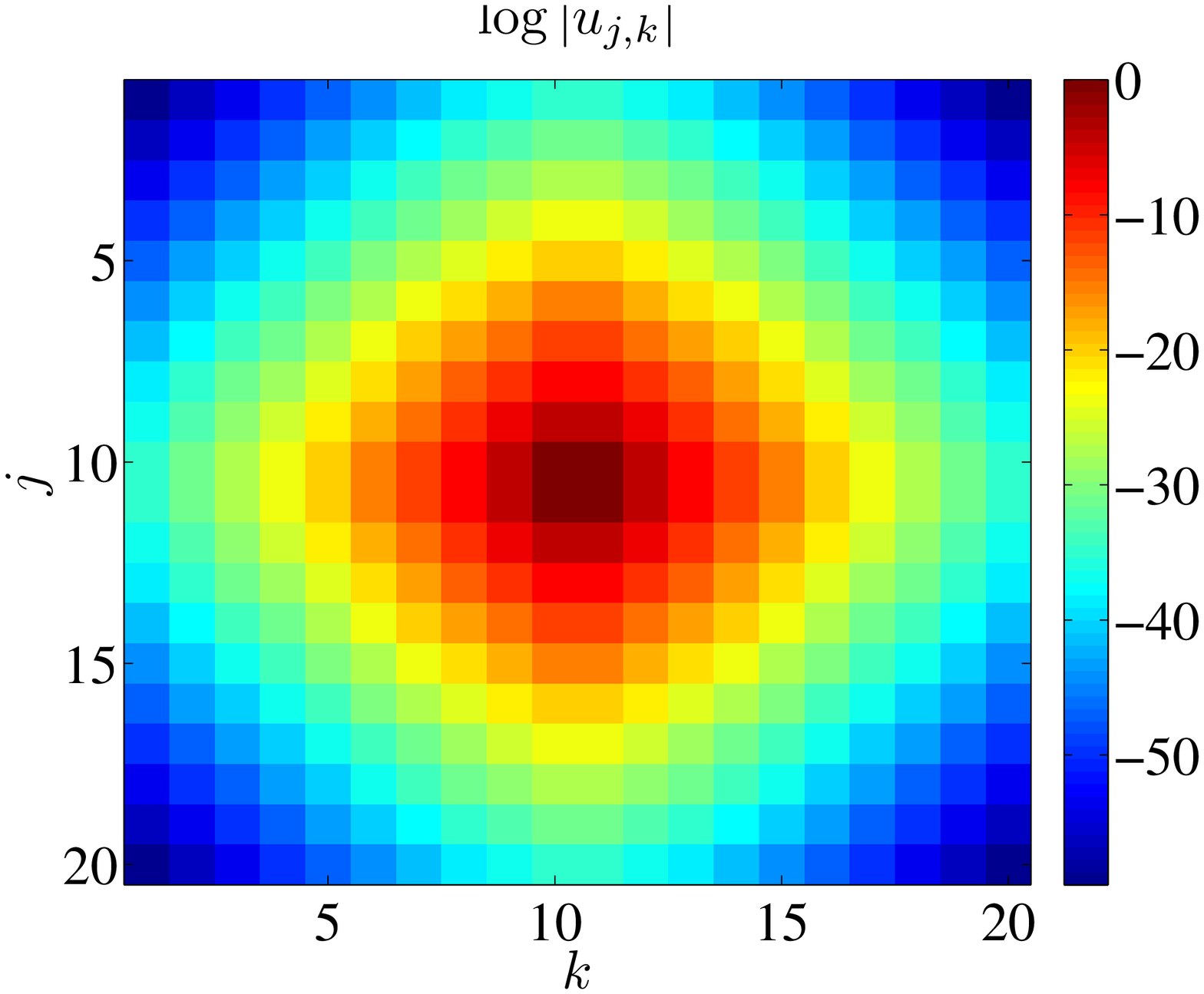}
\includegraphics[width=7cm]{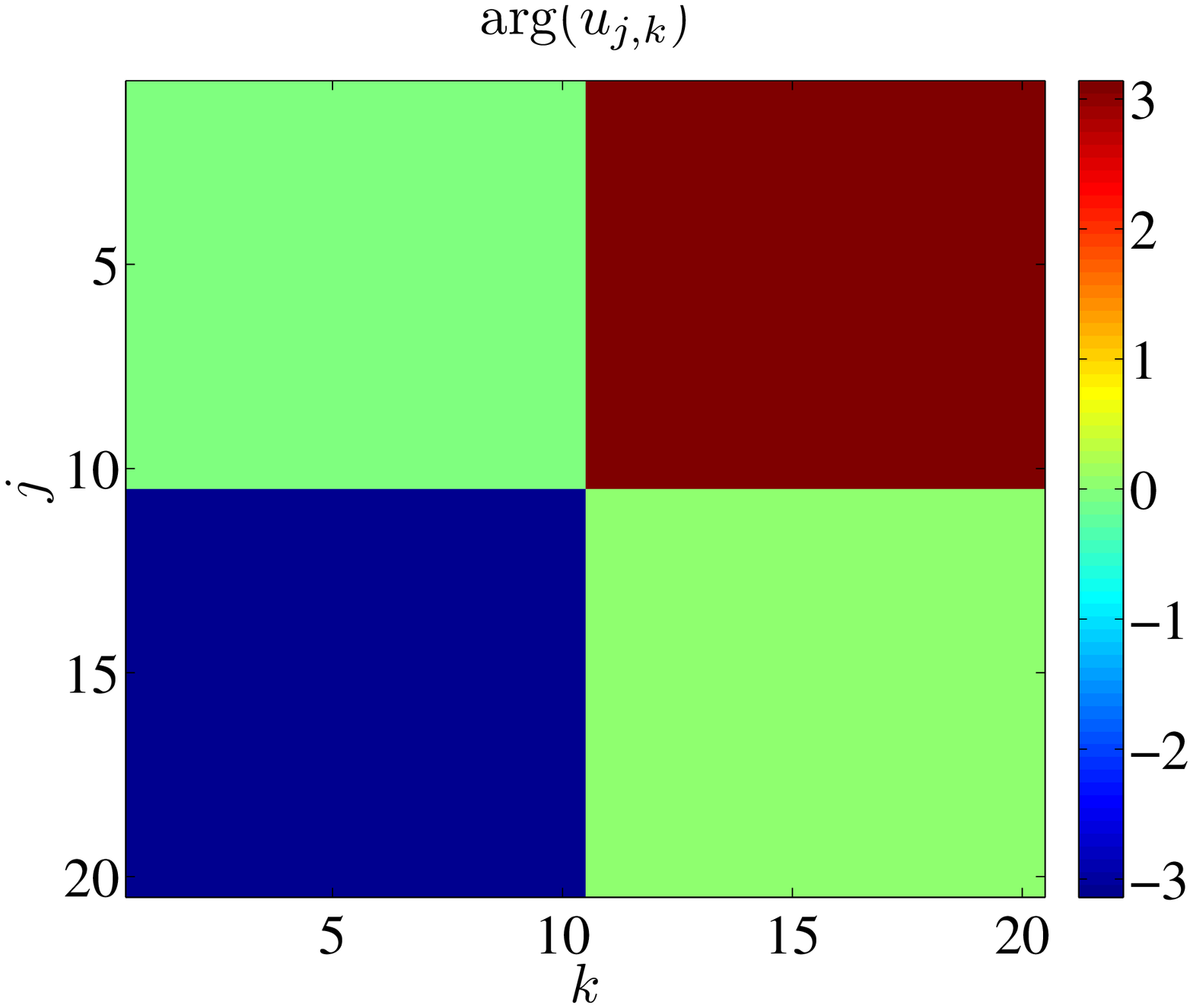}\\
\includegraphics[width=7cm]{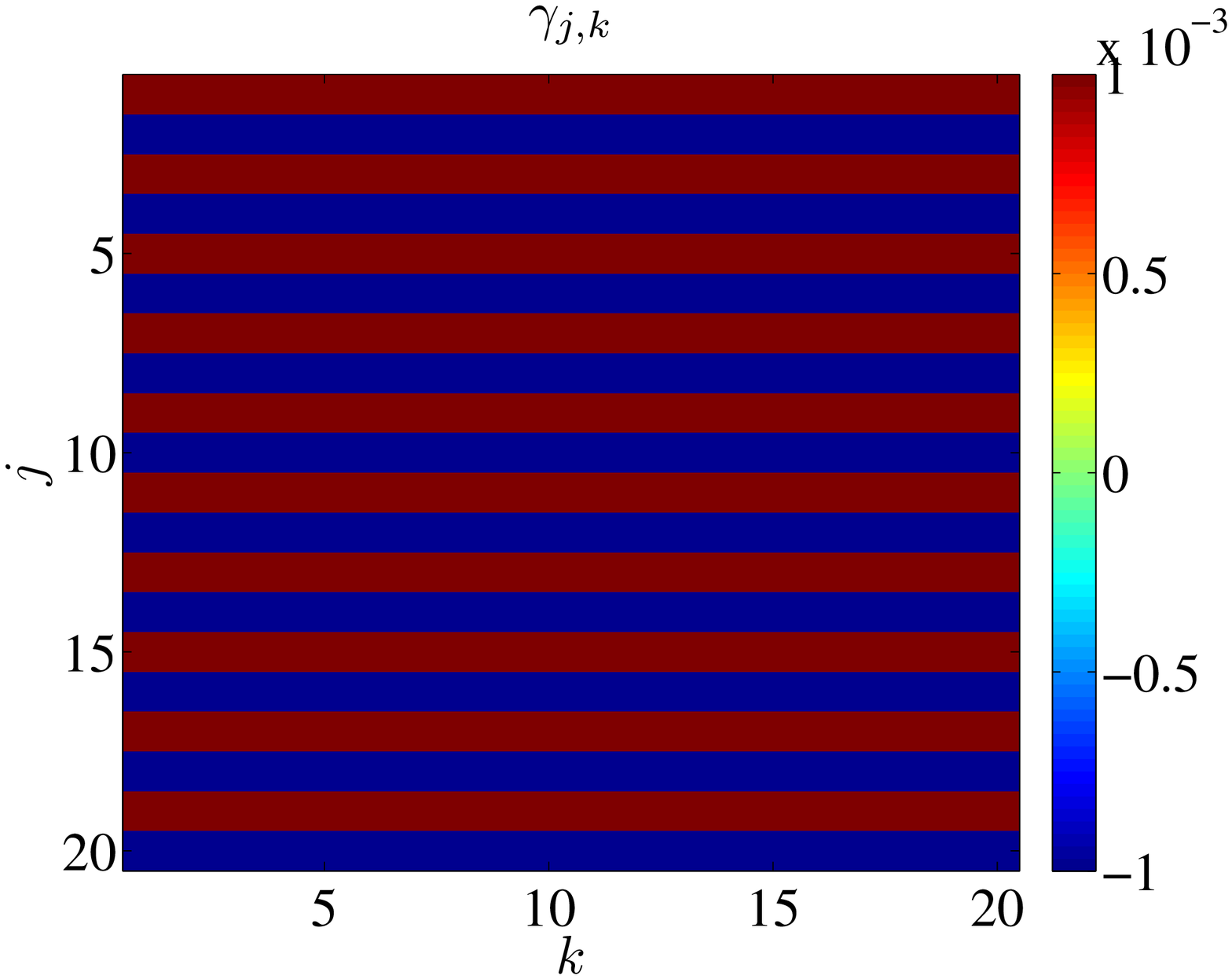}
\includegraphics[width=7cm]{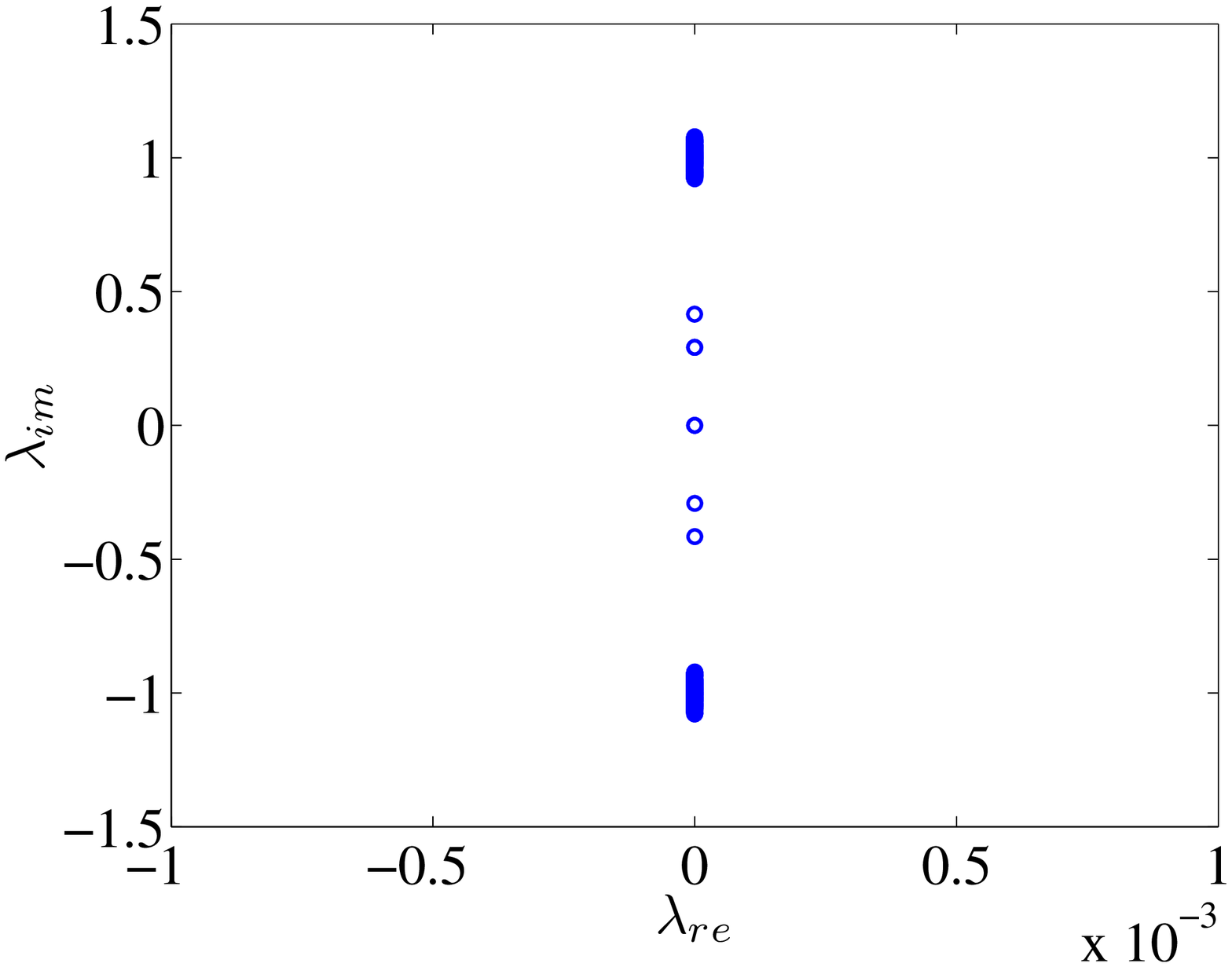}
\end{tabular}
\caption{An example of the branch (2-2-a) on the $20$-by-$20$ square lattice
with $\gamma_1=-\gamma_2=0.0001<\gamma_c(20)\approx 0.0003$ and $\epsilon=0.02$.
In the bottom right panel, we see eigenvalues $\lambda$ of the spectral problem (\ref{ev_matrix2})
are all on the imaginary axis, implying spectral stability of the stationary solution.}
\label{fig4_1}
\end{figure}

Finally, Figure \ref{fig4_2} provides some prototypical
examples of the numerical evolution of the branches (1-1-a) and (2-1-a),
which correspond to the $\pt$-symmetric configurations of Figure \ref{fig3_1}. Both
configurations are
unstable, therefore, the figure illustrates the development of these instabilities. In the (1-1-a) example (top panels),
the amplitudes of sites on $(11,11)$ and $(10,10)$ increase rapidly but those on $(10,11)$ and $(11,10)$ decrease.
In the (2-1-a) example (bottom panels), the situation is reversed and
sites on $(11,11)$ and $(11,10)$ grow rapidly while sites on $(10,10)$ and $(10,11)$ gradually decay.
Both of these time evolutions are also intuitive, as the
growth reflects the dynamics of the central sites bearing gain, while the
decay reflects that of the ones bearing loss.

\begin{figure}[!htbp]
\begin{tabular}{cc}
\includegraphics[width=7cm]{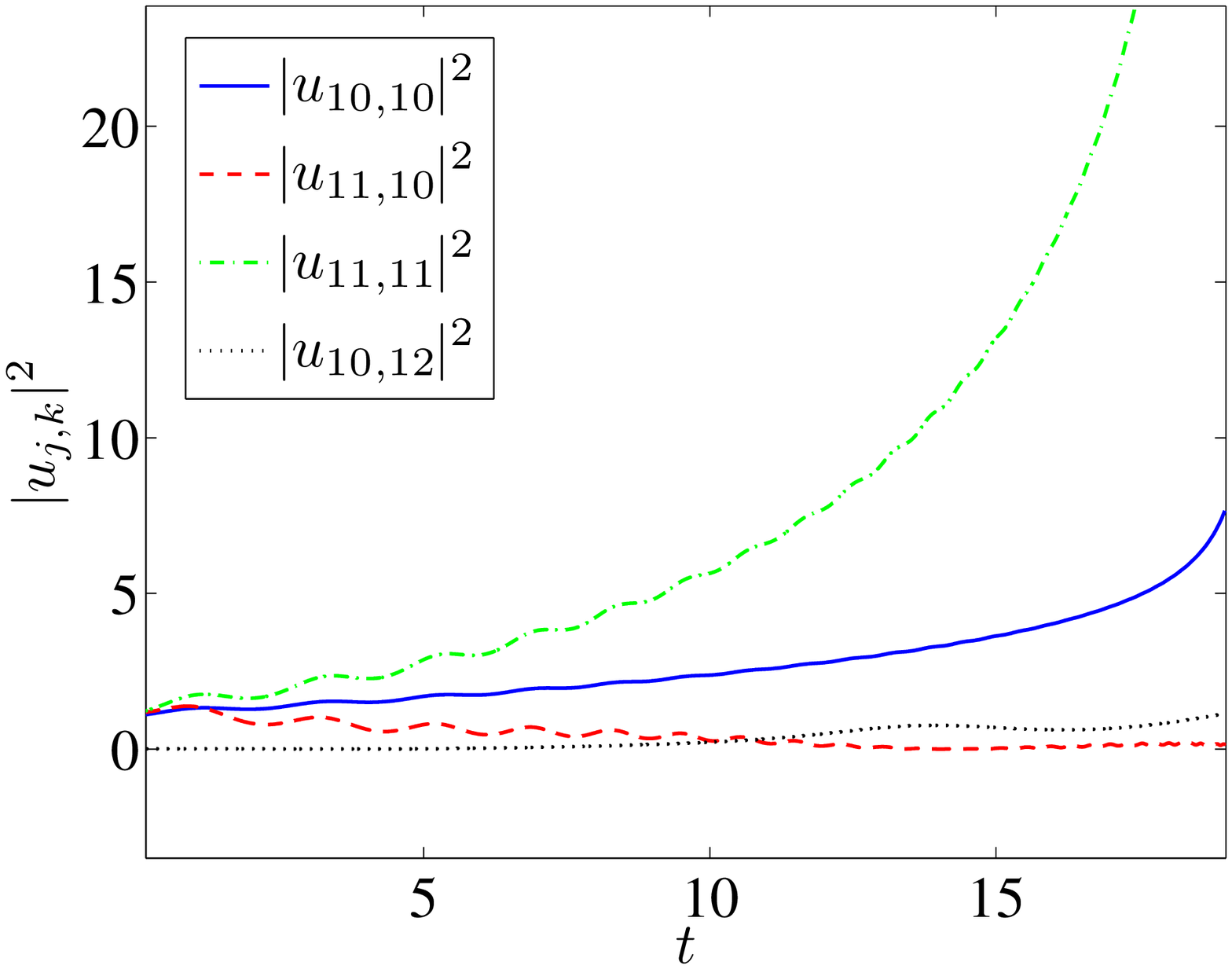}
\includegraphics[width=7cm]{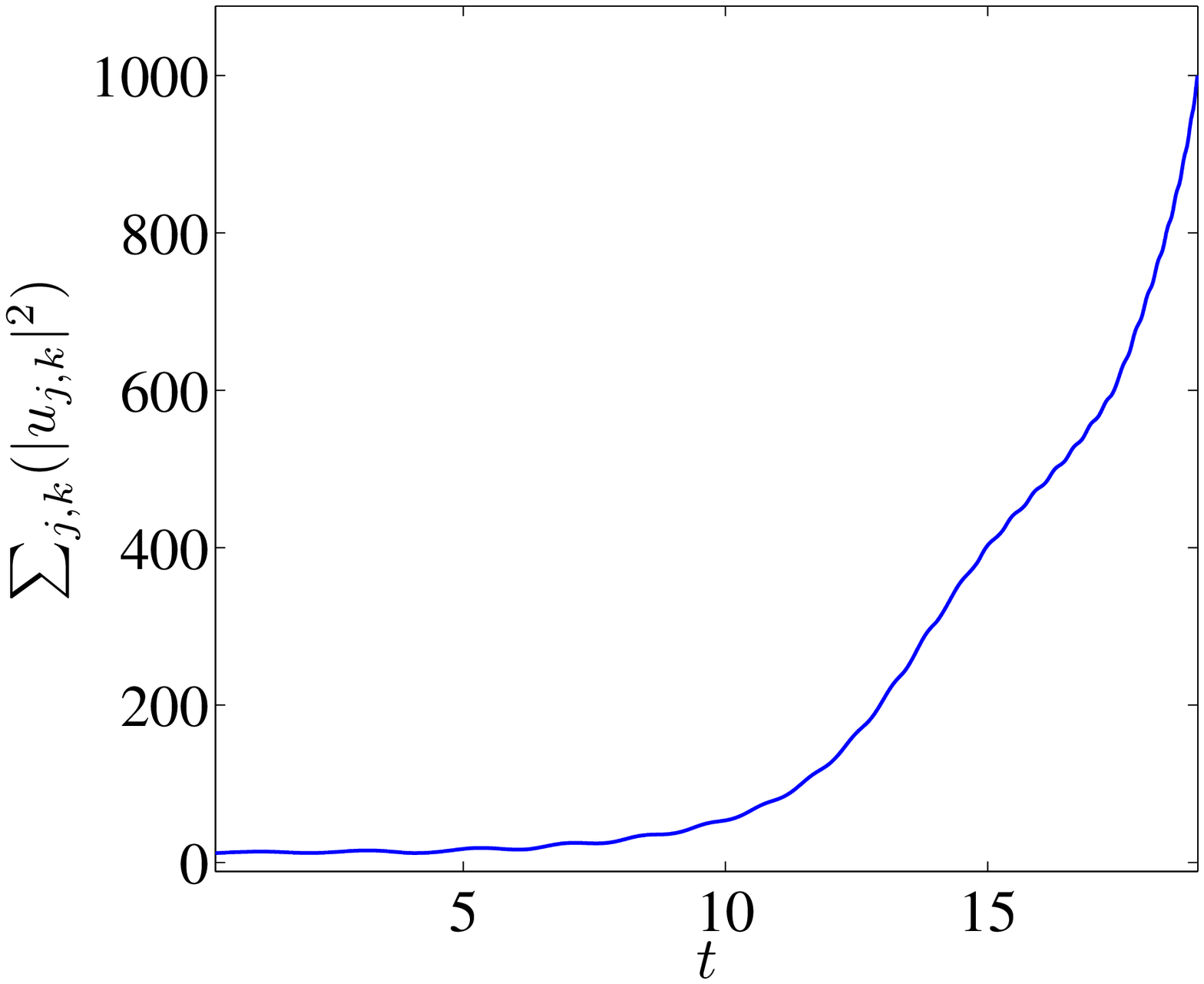}\\
\includegraphics[width=7cm]{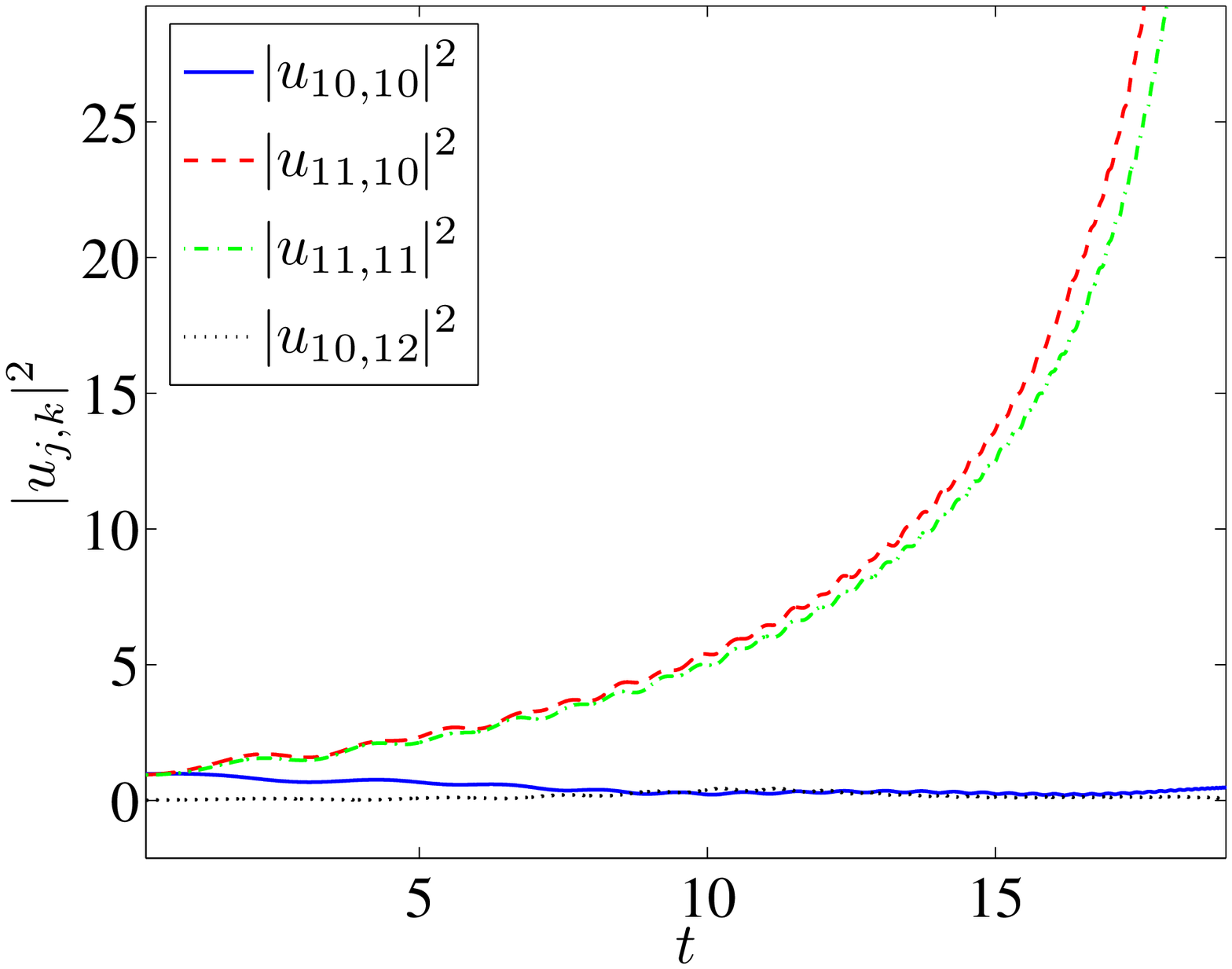}
\includegraphics[width=7cm]{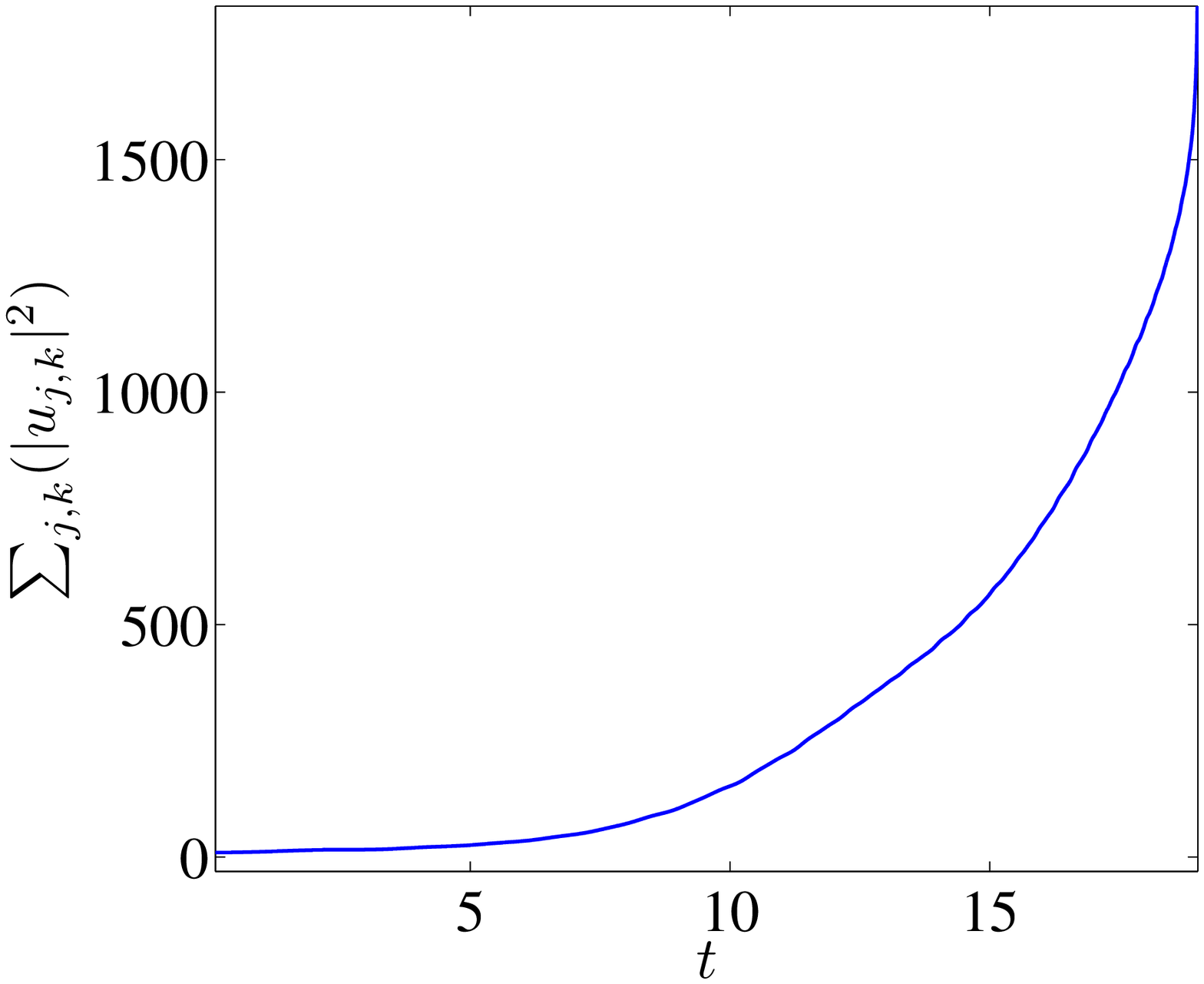}
\end{tabular}
\caption{The top panels show an example of the dynamics of the branch (1-1-a) on the 20-by-20 square lattice
with $\gamma_1=-\gamma_2=0.7$ and $\epsilon=0.1$. The bottom panels show similar dynamics for the branch
(2-1-a) with $\gamma_1=-\gamma_2=0.8$ and $\epsilon=0.1$.}
\label{fig4_2}
\end{figure}

\section{Conclusion}
\label{conclu}

In the present work, we have provided a systematic perspective on discrete soliton and vortex configurations
in the two-dimensional square lattices bearing $\pt$-symmetry. Both the existence and the stability features
of the $\pt$-symmetric stationary states were elucidated in the vicinity of a suitable anti-continuum limit,
which corresponds to the large propagation constant in optics. Interestingly, while discrete vortex solutions
were identified and a suitable family thereof could be continued
as a function of the gain-loss parameter $\gamma$, it was
never possible to stabilize the $\pt$-symmetric vortex configurations.
On the other hand, although states
extending discrete soliton configurations were found
generally to be also unstable, we found one exception of
generically stable (including under continuation) states.

This work paves the way for numerous additional explorations.
For instance, it may be relevant to extend the considerations
of the square lattice to those of hexagonal or honeycomb
lattices. Prototypical configurations in the $\pt$-symmetric settings
has been explored in~\cite{leykam}.
Moreover, it has been shown that the $\pt$-symmetric states may manifest unexpected
stability features, e.g., higher charge vortices
are more robust than the lower charge ones~\cite{terhalle}.
Another relevant possibility is to extend the present
consideration to a dimer lattice model, similarly to
what was considered e.g. in~\cite{malomchina}. The numerical
considerations of~\cite{malomchina} suggest that vortices
may be stable in suitable parametric intervals in such a setting.
Lastly, it may be of interest to extend the present
considerations also to three-dimensional settings,
generalizing analysis of the corresponding Hamiltonian model
of~\cite{kev3} (including cubes and diamonds)
and exploring their stability properties.

\vspace{1cm}

{\bf Acknowledgments:} P.G.K.~gratefully acknowledges the support of
NSF-DMS-1312856, as well as
the ERC under FP7, Marie Curie Actions, People, International Research Staff
Exchange Scheme (IRSES-605096). The work of D.P.
is supported by the Ministry of Education
and Science of Russian Federation (the base part of the
state task No. 2014/133).

\end{document}